\documentclass[runningheads,a4paper]{llncs}

\usepackage[margin=1.2in]{geometry}
\usepackage{makeidx}
\usepackage{amsmath}
\usepackage{mathtools}
\usepackage{graphicx}
\usepackage{tikz}
\usepackage{wrapfig}
\usepackage{bm}
\usepackage{todonotes}
\usepackage{verbatim} 
\usepackage{algorithm}[ruled,vlined,linesnumbered]
\usepackage{algpseudocode}
\usepackage{mathtools}
\usepackage{pgfplots}

\newcounter{phase}[algorithm]
\newlength{\phaserulewidth}
\newcommand{\setphaserulewidth}{\setlength{\phaserulewidth}}
\newcommand{\phase}[1]{%
  \vspace{-1.25ex}
  \Statex\leavevmode\llap{\rule{\dimexpr\labelwidth+\labelsep}{\phaserulewidth}}\rule{\linewidth}{\phaserulewidth}
  \Statex\strut\refstepcounter{phase}\textit{\thephase~--~#1}
  \vspace{-1.25ex}\Statex\leavevmode\llap{\rule{\dimexpr\labelwidth+\labelsep}{\phaserulewidth}}\rule{\linewidth}{\phaserulewidth}}
\makeatother

\setphaserulewidth{.7pt}

\newtheorem{prop}{Proposition}
\newtheorem{examp}{Example}

\usetikzlibrary{shapes, positioning, decorations.pathmorphing, decorations.markings}

\mathtoolsset{showonlyrefs}

\DeclarePairedDelimiter\abs{\lvert}{\rvert}
\makeatletter
\let\oldabs\abs
\def\abs{\@ifstar{\oldabs}{\oldabs*}}
\makeatother

\DeclarePairedDelimiter\set{\{}{\}}
\makeatletter
\let\oldset\set
\def\set{\@ifstar{\oldset}{\oldset*}}
\makeatother

\newcommand{\memupdate}{\mathcal{M}}

\newcommand{\real}{\bbbr}

\newcommand{\childset}{C}
\newcommand{\nonleaf}{N}
\newcommand{\commit}{S}
\newcommand{\leaf}{L}

\newcommand{\role}{\rho}

\DeclareMathOperator*{\argmin}{arg\,min}
\DeclareMathOperator*{\argmax}{arg\,max}

\begin{document}

\parindent=0pt
\baselineskip16pt
\parskip6pt
\parskip6pt
\titlerunning{SSE with Memory in Sequential Games}


\title{Computing Stackelberg Equilibrium with Memory in Sequential Games}
{\institute{Artificial Intelligence Center, Department of Computer Science, \\ Faculty of Electrical Engineering, Czech Technical University in Prague
}
\date{\today}

\author{Aditya Aradhye \thanks{Email address: a.aradhye@maastrichtuniversity.nl} \hskip6pt
Branislav Bo{\v s}ansk{\' y}\hskip6pt
Michael Hlav{\'a}{\v c}ek}}

\author{
        Aditya Aradhye
        \thanks{
        Corresponding author. Artificial Intelligence Center, Faculty of Electrical Engineering, Czech Technical University in Prague, Karlovo nám. 13, 120 00, Prague, Czech Republic. \\
        E-mail address:
        {\tt aradhadi@fel.cvut.cz}}
        \and
        Branislav Bo{\v s}ansk{\' y}
        \thanks{{\tt branislav.bosansky@agents.fel.cvut.cz}}
        \and
        Michael Hlav{\'a}{\v c}ek
        \thanks{{\tt michael@hlavacek.dev}.}
}

\maketitle

\begin{abstract}
Stackelberg equilibrium is a solution concept that describes optimal strategies to commit: 
Player~1 (\emph{the leader}) first commits to a strategy that is publicly announced,
then Player~2 (\emph{the follower}) plays a best response to the leader's commitment.
We study the problem of computing Stackelberg equilibria in sequential games with finite and indefinite horizons, when players can play history-dependent strategies. 
Using the alternate formulation called strategies with memory, we establish that strategy profiles with polynomial memory size can be described efficiently.
We prove that there exist a polynomial time algorithm which computes the Strong Stackelberg Equilibrium in sequential games defined on directed acyclic graphs, where the strategies depend only on the memory states from a set which is linear in the size of the graph. 
We extend this result to games on general directed graphs which may contain cycles. 
We also analyze the setting for approximate version of Strong Stackelberg Equilibrium in the games with chance nodes.
\end{abstract}




\section{Introduction}
The Stackelberg competition was introduced by von Stackelberg~\cite{von1934marktform} addressing an economic problem of duopoly. 
One player---\emph{the leader}---chooses a strategy to execute and publicly announces this strategy. The second player---\emph{the follower}---decides its own strategy only after observing the strategy of the first player. 
The leader must have commitment power (for instance a firm with monopoly in an industry) and cannot undo its publicly announced strategy, while the follower (for instance a new competitor) plays a best response to the leader's chosen strategy. 

The Stackelberg competition and the solution concept of a \emph{Stackelberg equilibrium} have been extensively studied in economics~(e.g., \cite{sherali1984multiple,AMIR19991,Matsumura2003}), in computer science~(e.g., \cite{von2004leadership,conitzer2006computing,letchford2010}), and have many important applications. One of the desired characteristics is that Stackelberg equilibrium \emph{prescribes} the leader a strategy to be executed, assuming the follower reacts to this strategy. This asymmetry arises in many real-world scenarios where the leader corresponds to a government or a defense agency that need to publicly issue and execute a policy (or a security protocol) which others can observe. 
Indeed, many of the successful applications of Stackelberg equilibrium strategies can be found in security domain~(see~\cite{tambe2011,sinha2018stackelberg} for an overview) or, for example, in wildlife protection~\cite{fang2016deploying,fang2017paws}.

We are interested in the computational questions of Stackelberg equilibrium and restrict to the two-player setting with one leader and one follower
\footnote{This setting has attained the most attention in the real-world applications. 
Moreover, computing a Stackelberg equilibrium with $2$ or more followers requires finding a specific Nash equilibrium in a general-sum game as a subproblem, which is already a computationally hard task~\cite{basilico2016methods}.}.
We focus on sequential (or dynamic) games that are played over time. 
In this setting, the leader commits to a randomized strategy in each possible decision point that can arise in the game and the follower plays a best response to this commitment.
Only few applied models consider strategic sequential interactions among the players directly~(e.g., in \cite{xu2018mitigating}), however, many currently modeled scenarios are essentially sequential -- police forces can commit to a security protocol describing not only their allocation to targets but also their strategy in case of an attack (or some other event) and results in a strategic response. 
Similarly, sequential models would allow rangers in national parks react to immediate observations made in the field.

In this paper, we study sequential games in which the state of the game is perfectly observable by the players. 
Games with determined finite horizon are played on a directed acyclic graph (DAG) and games without determined finite horizon are played on a general directed graph (DG). 
In such graphs, each node corresponds to a state of the game.
Each node is assigned to a player that decides which action (an edge) to take and the game transitions to another state.
Terminal states have assigned pair of utility values, one value for each player.
The goal of the leader is to find such a strategy to commit to that maximizes the expected outcome if the follower plays a best response.
Our solution concept is strong Stackelberg equilibrium (SSE), in which the follower breaks ties in favour of the leader. 
The complexity of computing equilibria in these games have been analyzed for many settings, however, most of the previous works is aimed at computing strategies which do not depend on past history (with the exception of works by Gupta et al.~\cite{gupta2015-conf}) -- i.e., the leader commits to a probability distribution over actions in each node assigned to her.

The determining result is due to Letchford and Conitzer~\cite{letchford2010}, showing a polynomial time algorithm to compute SSE in the sequential games with two players in the following settings: (1) The graphs are trees and players can play mixed strategies (2) The graphs are DAGs and players are restricted to play only pure strategies. 
The main idea in these algorithms is to compute the sets of the outcomes which the leader can guarantee, and these sets are used to construct commitment of the leader. 
They however, also show NP-hardness result of computing SSE in mixed strategies played on DAGs. 

The main driving factor of their hardness result is the following. In the directed acyclic graphs, there are large number paths arriving at a node. 
The leader needs to commit to one of many possible actions at that node, irrespective of the path which leads to this node.
If the number of such nodes is large (linear in the size of the graph), it might result in finding an optimal solution among the potentially exponentially many outcomes. 
This issue is exploited by \cite{letchfordthesis,letchford2010} to reduce the SAT problem to sequential games on DAGs with mixed strategies, thus proving the hardness result. 
Hence, the main issue here is the inability of the leader to make decisions irrespective of the path followed to reach a particular node. 

\subsection{Our contributions}

We address the natural question that arises from the above discussion.
Is it possible to compute the SSE in polynomial time for games on DAGs where players are allowed to use mixed strategies which can depend on history?  
Is there a way around to efficiently describe history-dependent strategies, since the possible number of histories is exponential in the size of the graph?
In this paper we positively answer both of these questions.

Since the players are rewarded utility only at the terminal nodes, the actual path taken to reach that node does not matter for the sake of computing the utility. 
However in games on DAGs and DGs, histories can still play an important role in determining the behaviour of the players, as they can be used as a virtual randomization device.
Using different histories to reach at the same node controlled by the follower, the leader can enforce certain behaviour from the follower in best response, and can also use the histories to punish the follower if the follower does not play accordingly.  
When players can use history-dependent mixed strategies, the solution of SSE differs from their history-independent counterpart. 
As the strategy space is richer, the leader is always (at least weakly) better off. Fig.\ref{dagexample} shows an example of a game in which strategy profile with memory gives a strictly better payoff for the leader when compared to the their history-independent counterpart.

In case of DAGs and DGs, there can be exponentially large number of histories. 
In such cases it is impossible to fully describe a general history-dependent strategy efficiently. 
However as we show that there are strategy profiles which can be described efficiently, as they uses only a partial information from the histories. 
To capture the idea that the players might not use the entire history space to make their decisions, we define a different formulation for history-dependent strategy profiles, called strategy profile with memory.
According this formulation, the strategy profiles are equipped with a set of public memory states. 
These memory states store a partial information about the past history.
The players can make decisions depending only on the current memory state, and not on the entire history.  
This formulation however does not restrict the set of history-dependent strategy profiles. 
If the set of memory states is equal to the set of all possible histories, then any arbitrary history-dependent strategy can be described as a strategy profile with memory, since both players are perfectly informed. 
On the other hand, if the set of memory states is small, particularly polynomial in size, then the strategies can be described efficiently.

The main result of this paper is the existence of polynomial time algorithm to compute SSE with memory (size of memory set is linear in  the size of graph) in DAGs. 
Firstly, the algorithm performs a dynamic programming style upward pass to constructs the commitment sets of the outcomes (probability distribution on terminal nodes) to which the leader can commit, that is, there is a commitment strategy for the leader and a corresponding best response from the follower such that the given outcome is reached. 
Also importantly, this upward pass also recursively stores the information which can be used later to compute the commitment for the leader.
Using this information, the algorithm then performs a downward pass, which constructs a commitment strategy for the leader. 
They crucial step in the downward pass is the construction of the memory states. 
The memory states act as a recommendation for the follower to preform certain actions so as to reach the outcome desired for the commitment of the leader.  
Whenever the follower takes an action other than the one suggested by the memory state, the game perpetually goes into 'red flag' memory state. 
On observing the red flag memory state, the leader plays punishing strategy which commits to the outcome with the lowest possible utility for the follower.
This acts as an incentive for the follower to act according the recommendations from the memory state. 

We extend our result to games on general DGs. 
Although the dynamic programming techniques can not be used directly due to the existence of directed cycles, we are able to modify our algorithm by decomposing the graph into strongly connected components (SCCs) and applying the dynamic programming on the set of SCCs.
The main driving factor here is that although the nodes may be visited multiple times in a SSE, which would make the computation difficult, we show an existence of a SSE in which any node is visited at most twice with positive probability.
This results allows us to use the upward pass twice to construct the commitment sets accurately, and the downward pass works similarly. 
When there are chance nodes, a reduction similar to ~\cite{letchford2010} shows that it is NP-hard to compute SSE with memory. 
\cite{bosansky2017} shows a polynomial time algorithm for the approximate version of SSE for history-independent strategies on trees.
With the strategies with memory formulation, we extend this result for DAGs.

\subsection{Related Work}
There is a rich body of literature studying the problem of computing Stackelberg equilibria. Our setting is closely related to the literature studying the Stackelberg equilibrium in the extensive form games. Letchford et al.\cite{letchford2010} showed the existence of polynomial time algorithm to compute the Stackelberg equilibrium in the extensive form games in the cases where either the graph was a tree or the players are restricted to play pure strategies. They also show the NP hardness result if there exist chance nodes or if the graph is DAG and strategies are mixed. 
Bo{\v s}ansk{\' y} et al.~\cite{bosansky2017} showed that allowing the leader to commit to correlated strategies can reduce computational complexity (e.g., computing Stackelberg equilibrium is NP-hard on games with chance, but polynomial when the leader commits to correlated strategies). They also show an FPTAS algorithm to compute the Stackelberg equilibrium in games with chance nodes.
Kroer et al.~\cite{kroer2018robust,kroer2020limited}, consider the Stackelberg setting in a limited lookahead variant of extensive form games.
In a related work for dynamic games, {\v C}ern{\' y} et al.~\cite{CernyEC20} proposed compact representation of strategies in the form of automata playing extensive-form games. In their work, the authors show that Stackelberg equilibrium can be computed in polynomial time when the complexity of the strategies is bounded. 

The main feature of our setting that separates our paper from the others is that in our setting, the players can base decision on the past history. To the best of our knowledge, no previous work analyzed the impact of allowing the leader to use memory in finite games played on directed graphs. 
There are some results for infinite stochastic games by Gupta et al.~\cite{gupta2015-conf} who considered Stackelberg equilibria with memory in games on directed graphs. However, their objective is to optimize discounted sum and they have negative results about finiteness of the memory needed. Later, these results were extended to mean payoff games~\cite{gupta2015}.

Our setting also relates to other class of Stackelberg games in a broader sense. The computational complexity of the problem is known for one-shot games~\cite{conitzer2006computing}, Bayesian games~\cite{conitzer2006computing}, and some infinite stochastic games~\cite{gupta2015,gupta2015-conf,letchford2012}.
Similarly, many practical algorithms are also known and typically based on solving multiple linear programs~\cite{conitzer2006computing}, 
or mixed-integer linear programs for Bayesian~\cite{Paruchuri2008} and extensive-form games~\cite{bosansky2015-aaai-sse,cermak2016-aaai}.
The Stackelberg equilibria are studied extensively in the various models of security games~\cite{Basilico2009,clempner2015stackelberg,durkota2019hardening,klaska2018,korzhyk2010complexity,nguyen2019tackling} 
Another variant of the Stackelberg notion is when the leader is allowed to commit to correlated strategies~\cite{Conitzer2011,letchford2012,xu2015}.

The structure of the paper is as follows. In section \ref{Sec:Prelim} we introduce the model and in the section \ref{Sec:Memory}, we define the concept of strategy profile with memory. In sections \ref{sec:algoDAG} and \ref{sec:algoDG} we discuss the polynomial time algorithm to compute the SSE in the games on DAGs and general DGs respectively. In section \ref{sec:algochanceDG}, we describe discuss the setting on directed acyclic graphs with chance nodes and in section \ref{sec:conclude}, we have concluding remarks.

\section{Sequential games}\label{Sec:Prelim}

We consider two-player general sum sequential games, both with finite and indefinite horizon. 

\begin{definition} \rm
A {\em two-player sequential game} is a tuple $\mathcal{G} = (G, v_0, \rho, \sigma_c, u)$ where
\begin{itemize}
\item $\set{1, 2}$ is a set of players, player $1$ is called the leader and player $2$ is called the follower,
\item $G = (V,E)$ is the directed graph on which the game is played, $V$ is the set of nodes and $E$ is the set of directed edges, 
\begin{itemize}
    \item $\nonleaf \subseteq V$ denotes the set of non-leaf nodes,
    \item $\leaf \subseteq V$ denotes the set of leaves,
    \item $\childset(v) \subseteq V$ denotes the set of children of node $v$, edges leading to the children nodes are also termed \emph{actions} at node $v$,
\end{itemize}
\item $v_0 \in V$ is the root node, where the game starts,
\item $\role : V \to \{1,2,c\}$ is a function which defines which player plays in the given node $s$, ($\role(v) = 1$ or $\role(v) = 2$), or whether the node is a chance node ($\role(v) = c$);
\begin{itemize}
    \item $V_i$ denotes the set of all states in which a player or the chance plays, $V_i = \set{v \in V \vert \role(v) = i}$
    \item An action of the players at any given node   corresponds to a child node they choose to move. Hence the action set at a node $v$ is same as the children set $C(v)$.
    $C_i = \bigcup_{v \in V_i} C(v)$ denotes the set of all actions of player $i$ or nature,
\end{itemize}
\item $\sigma_c: V_c \times V \to [0, 1]$ are the nature probabilities in chance nodes, $\sigma_c(v, w)$ denotes the probability with which the nature moves to $w$ starting at $v$. 
For each $v \in V_c$, $\sigma_c$ satisfies the following:  $\sum_{w \in C(v)} \sigma_c(v, w) = 1$
\item $u = (u_1, u_2)$ where $u_i : L \to \real$ is the utility function of each player
$i$.
\end{itemize}
\end{definition}

\textbf{Graph structure} The graph $G$ determines the horizon of the game.
The game graph is assumed to be a connected graph.
If the graph is a tree or a directed acyclic graph (DAG) then the game has finite horizon and if the graph has directed cycles then the game has infinite horizon. 
The graph does not contain self loops.
The \emph{Height} of the game graph which is a tree or a  DAG is the maximum distance from the root node to a leaf node. 
We use the term \emph{outcome} to denote the leaf node reached at the end of the play or a probability distribution on the set of leaf nodes reached if players mix their actions. 

\textbf{Histories} 
A {\em history} $h$ at a node $v$ is a sequence of nodes visited, formally $h = v_0,\ldots,v_k$ where $v_0$ is the root node, $v_k = v$ and for $0 \leq j \leq k-1$, $v_{j+1}$ is a child of $v_j$. 
For a given history $h$, let $v_h$ denote the final node in the sequence $h$.
For a node $v$, let $H_v$ denote the set of possible histories at $v$ and let $H = \cup_{v \in S}H_v$ be the set of all histories. Let $H_1 = \cup_{v \in S_1}H_v$ and $H_2 = \cup_{v \in S_2}H_v$ denote the set of histories for the leader and the follower respectively.
Any history of the form $v_0,\ldots,v_k,\ldots,v_l$ is a continuation of of history $v_0,\ldots,v_k$.
A \emph{play} $\pi$ is a sequence $v_0,\ldots,v_k$ where $v_0$ is the root node and $v_k$ is a leaf node, $\Pi$ is the set of all plays.

\textbf{Strategies and utilities} For $i=1,2$, a {\em history-dependent strategy} for player $i$ is defined as $\sigma_i : H_i \to \Delta(A)$ where 
$\sigma_i(h)$ denotes a probability measure over $A$ and
$\sigma_i(h, a)$ denotes a probability of an action $a \in A$ at history $h$. Each $\sigma_i$ satisfies the equality $\sum_{a \in A(v_h)} \sigma_i(h,a) = 1$.
Let $\Sigma_i$ denotes the set of history-dependent strategies of player $i$.
A {\em strategy profile} $\sigma = (\sigma_1, \sigma_2)$ is a tuple of strategies of the leader and the follower. 
Let $u_i(\sigma)$ denote the expected utility for player $i$ according to a strategy profile $\sigma$ and $u_i(\sigma)(h)$ denote the expected utility conditional on the history $h$ as been observed.
A play or a history has positive support in a strategy profile, if it is realized with positive probability.

\textbf{Strong Stackelberg equilibrium}
A strategy $\sigma_2$ is a {\em best response} to a strategy $\sigma_1$ if $u_2(\sigma_1, \sigma_2) \geq u_2(\sigma_1, \sigma'_2)$ for all $\sigma'_2 \in \Sigma_2$. 
Let ${\cal BR}(\sigma_{1})$ be the set of
all best responses to strategy $\sigma_1$.

\begin{definition}\label{def:SSE}
A strategy profile $\sigma$ is a {\em Strong Stackelberg equilibrium (SSE)} if
$$
\sigma = \argmax_{\sigma'_1 \in \Sigma_1, \sigma'_2 \in {\cal BR}(\sigma'_1)} u_1(\sigma'_1, \sigma'_2).
$$
\end{definition}

An important point to note here is that in a SSE, if the follower has more than one best response to the leader's strategy, he chooses a best response which maximizes the utility for the leader. 
This implies that, although there can be multiple SSEs, all of them provide the same utility for the leader. 

\vspace*{-5mm}

\begin{figure}
	\centering
	\begin{minipage}[t]{0.45\textwidth}
		\centering
		\begin{tikzpicture}
			[node distance=0.5cm]
			\tikzstyle{leader}=[circle,draw,inner sep=1.5]
  			\tikzstyle{leaf}=[]
		  	\tikzstyle{follower}=[rectangle,draw,inner sep=2]
  			\tikzstyle{nature}=[diamond,draw,inner sep=1.5]
  			\node(a) [follower]{$v_0$};
  			\node(b)[leader, below left=of a]{$v_2$};
  			\node(c)[leaf, below right=of a]{$(0, 1)$};
  			\node(d)[leaf, below left=of b]{$(0, 2)$};
  			\node(e)[leaf, below right=of b]{$(2, 0)$};
  			\draw (a) -- (b);
  			\draw (a) -- (c);
  			\draw (b) -- (d);
  			\draw (b) -- (e);
		\end{tikzpicture}
	\end{minipage}\hfill
	\begin{minipage}[t]{0.45\textwidth}
		\centering
		\begin{tikzpicture}
			[cross/.style={path picture={
  \draw[black]
(path picture bounding box.south east) -- (path picture bounding box.north west) (path picture bounding box.south west) -- (path picture bounding box.north east);
}}]
			\tikzstyle{dot}=[circle,draw,inner sep=1.5,fill=black]
  			\draw[->] (-0.5,0) -- (2, 0) node[right] {$u_2$};
		 	\draw[->] (0,-0.5) -- (0, 2) node[above] {$u_1$};
  			\node(a)[dot,label=left:{$(2,0)$}] at (0, 1.4) {};
  			\node(b)[dot,label=below:{$(0,2)$}] at (1.4, 0) {};
  			\node(c)[dot,label=below:{$(0,1)$}] at (0.7, 0) {};
			\node(i)[cross, inner sep=2, label=above right:{$(1,1)$}] at (0.7, 0.7) {};
  			\draw (a) -- (b);
  			\draw[draw=red,dashed] (c) -- (0.7, 1.6);
		\end{tikzpicture}
	\end{minipage}
	\caption{An example of an extensive form game. The leader plays in the circular node $v_2$, the follower
	in the rectangular $v_0$. To the right is the visualisation of possible outcomes in the space
	of utilities. The red dashed line represents the maxmin value $\mu_2(v_0)$
	of the follower's node $v_0$.\vspace{-0.3cm}}
	\label{stackelbergexample}
\end{figure}
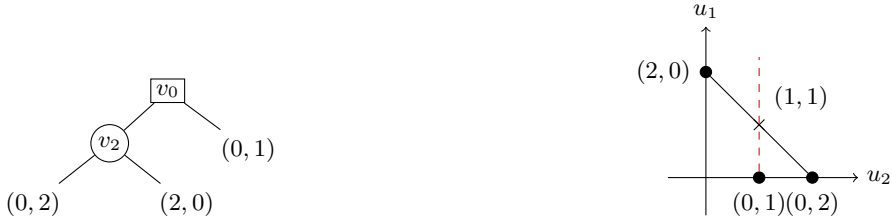

\begin{examp}\rm
To demonstrate SSE in sequential games, consider the game in Fig.~\ref{stackelbergexample}. In Nash equilibrium, the leader always
plays to $(2, 0)$ in $v_2$. Knowing this, the follower plays to $(0, 1)$ to get better utility.
In SSE, the leader has the option to commit to playing an even mixture between the leaves $(0, 2)$
and $(2, 0)$, offering both players the utility of $1$. Because the follower is assumed to break ties
in the leader's favor, she plays to $v_2$. Hence, the follower obtains better utility in the SSE.

In our results, we use a visualisation of outcomes of two player games due to \cite{letchford2010}.
We can identify a leaf node with
utilities $(u_1, u_2)$ with a point in a two dimensional space. A set of all mixtures
over two leaf nodes $z_1, z_2$ with utilities $(u_1^1, u_2^1)$ and $(u_1^2, u_2^2)$
is then a line segment connecting these two points. To illustrate this concept, refer
to Fig.~\ref{stackelbergexample}. The leader, who plays second,
has the option of mixing between two leafs with utilities $(2, 0)$ and $(0, 2)$,
hence the line segment. The vertical line positioned at $u_2 = 1$ represents the
follower's maxmin value -- the minimum utility the follower is willing to accept when playing Left.
\qed
\end{examp}

\subsection{Preliminary results} 

\subsubsection{}

\begin{definition}\rm
The {\em maxmin value at node $v$} for the follower is the worst-case utility the follower can guarantee for himself, irrespective of the leader's strategy. 
Formally, maxmin value at node $v$ is $\mu_2(v) = \max_{\sigma_2 \in \Sigma_2, h\in H_v}\min_{\sigma_1 \in \Sigma_1} u_2(\sigma_1, \sigma_2)(h)$. 
A {\em punishing strategy} is a memoryless strategy of the leader which for each $v \in V$ guarantees the payoff of at most $\mu_2(v)$ for the follower at node $v$. 
A punishing strategy $\sigma^p_1$ satisfies $u_2(\sigma^p_1, \sigma_2)(h) \le \mu_2(v)$ for each $\sigma_2 \in \Sigma_2$, $v \in V$ and $h\in H_v$ . 
\end{definition}

The maxmin value of the follower is completely independent of the leader's utility. While playing the punishing strategy, the leader acts as an adversary to the follower, disregarding their own utility. Maxmin value at any node is does not depend on through which history that node is reached. 

\begin{lemma}\label{lem:punishing}
For any sequential game on a graph which does not contain directed cycles, the leader has a punishing strategy.
\end{lemma}

\begin{proof}
We prove the Lemma by constructing the punishing strategy for the leader in the backward fashion. During this backward procedure, we also keep track of the maxmin value for the follower at each node $v$.

For a leaf node $v$, clearly $\mu_2(v)$ is equal to the payoff the follower obtains at the leaf $v$. Consider a node $v$ which is not a leaf. 
Assume that $\mu_2(w)$ has been calculated for each child $w$ of $v$. 
If $v$ is a leader node, then leader can choose the action which hurts the follower the most, hence $\mu_2(v) = \min_{w \text{ child of } v } \mu_2(w)$ and the leader moves to the node $arg\min_{w \text{ child of } v } \mu_2(w)$ with probability 1.
If $v$ is a follower node, then follower can minimize his punishment by moving to a child with largest maxmin value, hence $\mu_2(v) = \max_{w \text{ child of } v } \mu_2(w)$. 
\qed
\end{proof}

\begin{lemma}\label{lemm:follower-pure}
For any SSE $(\sigma_1, \sigma_2)$, there exists a SSE $(\sigma_1,\hat{\sigma_2})$ where $\hat{\sigma_2}$ is a pure strategy.
\end{lemma}

\begin{proof}
We will construct $\hat{\sigma_2}$ by modifying $\sigma_2$ at each history where $\sigma_2$ randomizes between more than one actions. 
Consider an arbitrary such history $h$, assume that at $h$, $\sigma_2$ plays actions $w_1,\ldots,w_k$ with positive probability. 
First we claim that $u_2(\sigma_1,\sigma_2)(hw_i)=u_2(\sigma_1,\sigma_2)(hw_j)$ for each $i,j=1,\ldots,k$. 
Assume the contrary, that for some $i,j$,  $u_2(\sigma_1,\sigma_2)(hw_i)>u_2(\sigma_1,\sigma_2)(hw_j)$. 
Then follower can obtain a strictly better payoff if he deviates from $\sigma_2$ by transferring all the weight of action $w_j$ to action $w_i$. 
Note that, in the profile $(\sigma_1,\sigma_2)$ the follower plays both $w_i$ and $w_j$ with positive probability, so the leader can not distinguish the strategy $\sigma_2$ from its deviation just by observing the history and thus can not punish the deviation. 
This contradicts that $\sigma_2$ is a best response to $\sigma_2$. So, we must have $u_2(\sigma_1,\sigma_2)(hw_i)=u_2(\sigma_1,\sigma_2)(hw_j)$ for each $i,j=1,\ldots,k$. 
Let $i_0 = \argmax_i\{u_1(\sigma_1,\sigma_2)(hw_i)\}$ and let $\hat{\sigma_2}$ plays $i_0$ with probability 1 at $h$.
Since $u_1(\sigma_1,\hat{\sigma_2})(h) \ge u_1(\sigma_1,\sigma_2)(h)$, $(\sigma_1,\hat{\sigma_2})$ is also a SSE. 
By performing the modification of $\sigma_2$ at every history, we obtain a pure strategy $\hat{\sigma_2}$ such that $(\sigma_1,\hat{\sigma_2})$ is a SSE.
\qed
\end{proof}

\begin{lemma}\label{lemm:onlytwo}
Consider an outcome in which payoff vectors $(x_1,y_1),\ldots,(x_k,y_k)$ are realized with probabilities $p_1,\ldots,p_k$. Then there exists 2 vectors, WLG say $(x_1,y_1), (x_2,y_2)$ and a probability $\hat{p}$, such that the outcome in which $(x_1,y_1)$ is realized with probability $\hat{p}$ and $(x_2,y_2)$ is realized with probability $(p_1+\ldots+p_k) - \hat{p}$ is weakly preferred by both the players.       
\end{lemma}

\begin{proof}
Let $p = p_1+\ldots+p_k$. 
We view the payoff vectors as points in 2-dimensional plane.
The point $A = p_1(x_1,y_1)+\ldots+p_k(x_k,y_k)$ is a convex combination of the points $p(x_1,y_1),\ldots,p(x_k,y_k)$, hence lies in the interior or on the boundary of their convex hull. 
So there exists a point $B$ on the boundary of the convex hull with same payoff for the leader and (possibly weakly) better payoff for the follower than $A$.

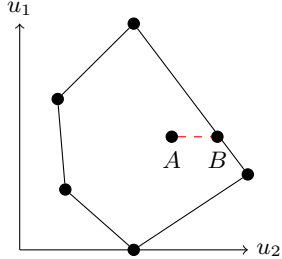
\begin{figure}
	\centering
	\begin{tikzpicture}
			[cross/.style={path picture={
  \draw[black];}}]
			\tikzstyle{dot}=[circle,draw,inner sep=1.5,fill=black]
  			\draw[->] (0,0) -- (3, 0) node[right] {$u_2$};
		 	\draw[->] (0,0) -- (0, 3) node[above] {$u_1$};
  			\node(a)[dot] at (0.6, 0.8) {};
  			\node(b)[dot] at (1.5, 0) {};
  			\node(c)[dot,label=below:{$A$}] at (2, 1.5) {};
  			\node(d)[dot] at (3, 1) {};
  			\node(e)[dot] at (1.5, 3) {};
  			\node(f)[dot] at (0.5, 2) {};
  			\node(g)[dot,label=below:{$B$}] at (2.6, 1.5) {};
  			\draw (a) -- (b);
  			\draw (a) -- (f);
  			\draw (e) -- (f);
  			\draw (e) -- (d);
  			\draw (d) -- (b);
  			\draw[draw=red,dashed] (c) -- (g);
\end{tikzpicture}
\end{figure}
\qed
\end{proof}

\begin{prop}
In any SSE with history-dependent strategies, the leader obtains (weakly) better payoff than any SSE with history-independent strategies.
\end{prop}

\vspace{-0.3cm}
In Section \ref{Sec:Memory} we show an example of a game in which the leader obtains strictly better payoff in SSE with history-dependent strategies, than SSE with history-independent strategies.

\vspace{-0.3cm}
\section{Strategy profiles with Memory}\label{Sec:Memory}

When the graph is not a tree, the number of different histories can be exponentially large. 
So the history-dependent strategies can not be described efficiently.
In this section, we define a different formulation for history-dependant strategies profiles, called \emph{Strategy profiles with Memory}. 
In this formulation, the profiles are equipped with a set of memory states which store partial information about the past histories. 
Players do not make decisions according to a complete history, but only according to a memory state. 
We assume that the memory is public, that is, both the leader and the follower observe the same memory. 


\begin{definition}A {\em strategy profile with memory} is a tuple $(\sigma, M, \memupdate, m_0)$, where
\begin{itemize}
\item $M$ is a set of memory states
\item $\sigma = (\sigma_1, \sigma_2)$ are the strategies of players with $\sigma_i: M \times V_i \times V \to [0, 1]$ where $\sigma_i(m, v, w)$ denotes the probability that player $i$ at node $v \in V_i$ moves to a node $w \in C(v)$.
For each $v \in V_i$, $\sigma_i$ satisfies the following:  $\sum_{w \in C(v)} \sigma_c(m, v, w) = 1$.
\item $\memupdate : M \times V \times V \to M$ is the memory update function,
where $\memupdate(m, v, w)$ is the updated memory state if a player or nature moves to node $w$ from node $v$ when the current memory state is $m$
\item $m_0 \in M$ is the initial memory state.
\end{itemize}
\end{definition}


If the set of memory states $M$ is equal to the set of all possible histories $H$, then a general history-dependant strategy profile can be described using this formulation.   
However, when the set $M$ is at most polynomial in size, the strategy profiles can be described efficiently. 
A strategy profile with history-independent strategies can be described by choosing $M$ to be a singleton set.

\subsection{Memory as a randomization device}

In this subsection we discuss the use with memory as a randomization device to obtain better expected payoff for the leader.
We explain this concept using a simple example which captures the main idea.    

\vspace*{-0.5cm}
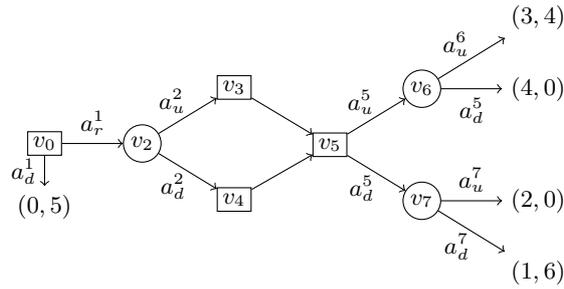
\begin{figure}
\centering
\begin{tikzpicture}
		[node distance=0.4cm and 0.8cm, label distance=-8pt]
	  \tikzstyle{leader}=[circle,draw,inner sep=1.5]
   	\tikzstyle{leaf}=[]
 	  \tikzstyle{follower}=[rectangle,draw,inner sep=2]
  	\tikzstyle{nature}=[diamond,draw,inner sep=1.5]
  	\node(a)[follower]{$v_0$};
  	\node(b)[leader, right=of a]{$v_2$};
  	\node(c)[leaf, below=of a]{$(0,5)$};
  	\node(d)[follower, above right=of b]{$v_3$};
  	\node(e)[follower, below right=of b]{$v_4$};
  	\node(g)[follower, above right=of e]{$v_5$};
  	\node(i)[leader, above right=of g]{$v_6$};
  	\node(j)[leader, below right=of g]{$v_7$};
  	\node(k)[leaf, right=of i]{$(4, 0)$};
  	\node(l)[leaf, above=of k]{$(3, 4)$};
  	\node(m)[leaf, right=of j]{$(2, 0)$};
  	\node(n)[leaf, below=of m]{$(1, 6)$};
  	\draw[->] 
		  (b) edge node[label=above left:{$a^2_u$}]{} (d)
		  (b) edge node[label=below left:{$a^2_d$}]{} (e)
          (d) edge (g) 
          (e) edge (g)
		  (g) edge node[label=above left:{$a^5_u$}]{} (i)
		  (g) edge node[label=below left:{$a^5_d$}]{} (j)
		  (i) -- node[midway,sloped,below]{$a^5_d$} (k)
		  (i) edge node[label=above left:{$a^6_u$}]{} (l)
		  (j) edge node[label=below left:{$a^7_d$}]{} (n);
		  \draw[->] (j) -- node[midway,sloped,above]{$a^7_u$} (m);
		  \draw[->] (a) -- (b) node[midway,sloped,above]{$a^1_r$};
		  \draw[->] (a) -- (c) node[midway,left]{$a^1_d$};
\end{tikzpicture}
\caption{An example of a game on a DAG. The leader plays in circular nodes,
the follower plays in square nodes.}
\label{dagexample}
\end{figure}

\begin{examp}\label{Ex:memory} \rm
Consider the game described by the graph in Fig.\ref{dagexample}, where $v_0$ is the root node. 
The SSE with history-dependant strategies can be described as follows. 
The leader wants to make the follower play $a_r^1$ in $v_0$, in order to avoid obtaining utility $0$. 
So at node $v_2$, the leader must commit to a randomized strategy in order to offer the follower utility at least $5$. 
If the leaf nodes $(3, 4)$ and $(1, 6)$ both are reached with equal probabilities, then the follower is guaranteed the expected payoff $5$ and the leader gets expected payoff $2$. 
In order to do so, the leader plays the actions $a_u^2$ and $a_d^2$ with equal probabilities at node $v_2$ and conditions his strategy in $v_6$ and $v_7$ on whether $a^2_u$ or $a^2_d$ was played. 
If $a^2_u$ was played, the leader wants to reach the leaf $(3, 4)$, and so hopes that the follower plays $a_u^5$ at $v_5$ and if $a^2_d$ was played, the leader wants to reach the leaf $(1, 6)$, and so hopes that the follower plays $a_d^5$ at $v_5$. 
He incentivizes the follower to play $a_u^5$ at $v_5$ on observing history $a^2_l$ at $v_2$, by committing $a_u^6$ at $v_6$ and $a_u^7$ at $v_7$, (so the follower is better off playing  $a_u^5$ at $v_5$, as the follower is punished by with utility 0 if he plays $a^5_r$).
Symmetrically, if $a^2_r$ was played, the leader commits to playing to $(1, 6)$ in $v_7$ and to $(4, 0)$ in $v_6$.

The memory states  for this SSE can be described as follows. 
The set $M$ has 3 memory states, $M = \set{m_0, m_l, m_r}$ where $m_0$ is initial memory state. 

The above analysis relies on the key fact that both the players can play depending on the past history. Otherwise for instance, if the leader's strategy is history-independent, he would not be able to punish the follower. And so in the best response, the follower plays $a_d^7$ irrespective of history to obtain himself the utility 6 and the leader gets the utility 1. Hence the leader obtains strictly better payoff in SSE with history-dependent strategies, than SSE with history-independent strategies.
\end{examp}

\section{Algorithm for Games without Chance Nodes on a DAG}\label{sec:algoDAG}

In this section, we consider the games on DAGs without chance nodes. 
As there are no directed cycles, the game has finite horizon. 
The main result is the construction of a polynomial time algorithm that computes a SSE where the players use strategy profile with memory. 
We show that the size of memory set is linear, thus it can be described very efficiently. 

\begin{theorem}\label{thm:DAG}
In games on DAGs without chance nodes, there exists a SSE with strategy profile with memory, where the size of memory set is linear in the number of nodes. There exists a polynomial time algorithm which computes such a SSE. 
\end{theorem}

Our algorithm consists of two parts. 
In the first part, the algorithm performs an upward dynamic programming pass, which constructs for each node $v$ the \emph{commitment set}, which is the set of all the outcomes to which the leader can commit.
If leader can commit to number of different outcomes, he can also commit to any convex combination of those outcomes by choosing the appropriate mixed strategy.
Hence the commitment set is convex and infinite.
As the commitment set is infinite, it is impossible store all its points. 
However, the convexity allows an efficient representation of $\commit_v$, the algorithm stores only finitely many points whose convex hull is the commitment set.
We denote this finite set by $S_v$.
In the upward pass, the algorithm also constructs a function $S$, which for every node $v$ and outcome $p$ in $\commit_v$, defines the route with which leader can commit to $p$ starting at $v$.

In the second part, the algorithm performs a downward pass, which constructs a commitment strategy $\sigma$ with memory for the leader, the memory update function $\memupdate$ and the set of memory states $M$.
The labels constructed in the upward pass are used to construct the commitment at each node. 

\vspace{2mm}

\begin{algorithm}
  \caption{\textbf{Upward Pass}}\label{algp:up}
  \begin{algorithmic}[1]
    \State Sort $V$ in reverse topological order
    \For{\texttt{$v$ in $V$}}
        \State $S_v \leftarrow \{\}$
        \If{\texttt{$v$ is leaf node}}
        \textsc{UpwardLeaf(v)}
        \ElsIf{\texttt{$v$ is a leader node}}
        \textsc{UpwardLeader(v)}
        \Else 
        \textsc{UpwardFollower(v)}
        \EndIf
    \EndFor
\phase{Leaf nodes}    
     \Procedure{UpwardLeaf}{$v$}
        \State Add $v$ to $S_v$
        \State $L(v,v) \leftarrow ()$
    \EndProcedure
\phase{Leader nodes}    
    \Procedure{UpwardLeader}{$v$}
        \For{\texttt{$w$ in $C(v)$}}
            \For{\texttt{$p$ in $S_w$}}
                \State Add $p$ in $S_v$
                \State $L(v,p) \leftarrow (1,w,p)$ 
            \EndFor
        \EndFor
    \EndProcedure
\phase{Follower nodes}    
    \Procedure{UpwardFollower}{$v$}
        \For{\texttt{$w$ in $C(v)$}}
            \For{\texttt{$p$ in $S_w$}}
                \If{\texttt{$u_2(p) \ge i(w)$}}
                    \State Add $p$ to $S_v$ 
                \EndIf
                \State $L(v,p) \leftarrow (1,w,p)$
            \EndFor
            \For{\texttt{$p_1$ and $p_2$ in $S_w$}}
                \If{\texttt{$\alpha \cdot u_2(p_1) + (1-\alpha) \cdot u_2(p_2) = i(w)$ for $0<\alpha<1$}}
                    \State Add $p=\alpha \cdot u_2(p_1) + (1-\alpha)$ to $S_v^w$ 
                \State $L(v,p) \leftarrow (\alpha,w,p_1,1-\alpha,w,p_2)$
                \EndIf
            \EndFor
            \State Add outcome in $S_v^w$ with highest utility for leader to $S_v$  
        \EndFor
    \EndProcedure
  \end{algorithmic}
\end{algorithm}

\noindent {\bf Upward pass}\space

\vspace{1mm}

Each outcome can be stored as a point in two-dimensional space, whose $x$-coordinate denotes the follower's expected utility and $y$-coordinate denotes the leader's expected utility (see Fig. \ref{stackelbergexample}).
For each node $v$ and outcome $p$ in $S_v$, algorithm maintains the label $L(v,p)$, which stores the information about which mixed path from $v$ the leader can guarantee in order to obtain the outcome $p$.
Each label is stored as am ordered tuple.
This label will be used in Part 2 to create a procedure which computes the commitment strategy for the leader.
The label is stored in the form of an ordered tuple, where an element is added for each action which is suggested positive probability.  

The upward pass starts by fixing a reverse topological order of the nodes, the commitment sets for each node are constructed in that order.
This construction is easy for leaf nodes, since starting at a leaf node, the leader can commit to a unique outcome. 
For non-leaf nodes, they are constructed using the commitment sets already constructed for their children. 

\noindent\underline{Leaf nodes:} 

For each leaf node $v \in \leaf$, algorithm assigns $S_v = \{v\}$. For $v$ in $S_v$, assigns an empty label $L(v,v) = ()$.

\underline{Leader nodes:} 

Consider a leader node $v$, a child node $w \in \childset(v)$ and an outcome $p \in \commit_w$, so leader can commit to the outcome $p$ starting at $w$. 
Then leader can also commit to $p$ from $v$, by simply committing to go to $w$ at $v$. 
Hence, for each $p \in S_w$, we have $p \in S_v$.
For each $p \in S_w$, the algorithm assigns 
$L(v,p) = (1,w,p)$
The label $(1,w,p)$ for committing to outcome $p$ starting at $v$ can be read as an instruction as follows: First commit at $v$ to $w$ with probability 1 and then commit to outcome $p$ starting at $w$ using the label $L(w,p)$.
Such a label is called a pure label.

Leader can also commit by playing a mixed action at $v$. 
By Lemma \ref{lemm:onlytwo} for purpose of computing a SSE, it is sufficient to only consider the outcomes obtained by mixing over at most two children. 
Consider children $w_1,w_2 \in \childset(v)$ and outcomes $p_1 \in \commit_{w_1}$, $p_2 \in \commit_{w_2}$. 
Then leader by committing to $w_1$ and $w_2$ with appropriate probabilities can also commit to any outcome which is convex combination of $p_1$ and $p_2$. 
However since $p_1, p_2$ are already in $S_v$, and the algorithm stores only the extreme points of the convex set $S_v$, the algorithm does not store the convex combinations of $p_1$ and $p_2$. 
Hence, $S_v = \Big(\cup_{w\text{ child of }v} S_w\Big)$.

\vspace{1mm}
\noindent\underline{Follower nodes:} The case with follower nodes is bit more complicated for the following reason.
Consider a follower node $v$, a child node $w \in \childset(v)$ and an outcome $p \in \commit_w$, so leader can commit to the outcome $p$ starting at $w$.
Then it is not necessarily true that leader can also commit to $p$ from $v$.
In fact, the leader can commit to $p$ from $v$, if and only if the follower has an incentive to go to $w$ from $v$. 
To prevent the follower from going to node $w'$ instead of $w$, the leader can commit at $w'$ to the outcome in $S_{w'}$ which gives the follower least utility.
If $u_2(p) \ge \min_{p'\in S_{w'}} u_2(p')$ then the follower obtains higher expected utility by going to $w$ at which the leader commits to $p$, than going to $w'$ at which the leader commits to $\argmin_{p'\in S_{w'}} u_2(p')$. 
So the follower has no incentive to prefer $w'$ over $w$. 
Denote the quantity $i(w) = \max_{w' \neq w}\min_{p'\in S_{w'}} u_2(p')$ where $w' \in C(v)$. 
Hence the algorithm adds $p$ in $S_v$ if $u_2(p) \ge i(w)$. 
The label $L(v,p)$ is defined as $L(v,p) = (1,w,p)$.

Now consider two arbitrary outcomes $p_1,p_2$ in $S_w$ such that $u_2(p_1) \ge i(w)$ and $u_2(p_2) < i(w)$.
Then there exists an $\alpha \in (0,1]$ such that $p = \alpha p_1 + (1-\alpha) p_2$ and $u_2(p) = i(w)$. 
In this case, the leader is not able to commit to $p_2$, but he is able to commit to $p$ by first committing to $w$ and then at $w$, committing to outcome $p_1$ with probability $\alpha$ and to outcome $p_2$ with probability $1-\alpha$. 
Let $S_v^w$ a set of such outcomes $p$. 
Note that for any $p,p' \in S_v^w$, $u_2(p) = u_2(p')$.  
The algorithm adds for each $w \in C(v)$, one outcome in $S_v^w$ (provided $S_v^w$ is not empty) which has the maximum utility for the leader (in case of multiple outcomes with same utility for the leader, one outcome is arbitrarily chosen). 

If $p \in S_v^w$ is added to $S_v$, the label $L(v,p)$ is defined as $L(v,p) = (\alpha,w,p_1,1-\alpha,w,p_2)$. 
The label $(\alpha,w,p_1,1-\alpha,w,p_2)$ for committing to the outcome $p$ starting at $v$ can be read as an instruction as follows: 
With probability $\alpha$ commit at $v$ to $w$ and then commit to outcome $p_1$ starting at $w$ using the label $L(w,p_1)$ and with probability $1-\alpha$ commit at $v$ to $w$ and then commit to outcome $p_2$ starting at $w$ using the label $L(w,p_2)$.
Such a label is called a mixed label.

\noindent \underline{Solution:} 

After constructing the commitment sets for all the nodes, the algorithm computes the outcome in $S_{v_0}$ with maximum $y$-coordinate (recall that $v_0$ is the initial node).  
This is the outcome the leader can commit starting from $s_0$ which gives him the maximum expected utility.
Formally, $p^* = \argmax_{p \in S_{v_0}} u_1(p)$. 
\qed

\vspace{3mm}

\noindent{\bf Downward pass}\space 

\vspace{1mm} During the downward pass, the algorithm constructs the commitment strategy $\sigma$ of the leader, the set of memory states $M$ and the memory update function $\memupdate$. 
The procedure $\textsc{Strategy}(v,p,m)$ at a node $v$ with a target outcome $p$ and current memory state $m$ performs the following 3 steps, 
\begin{enumerate}
    \item \underline{Memory update:} After a player moves to a node $w$ from $v$ on observing the memory $m$, this step updates the memory state to a new state denoted by $\memupdate(m,v,w)$.  
    \item \underline{Commitment of the leader:} If $v$ is a leader node, this step construct a (possibly) mixed actions at $v$ for committing to the target outcome $p$.
    If $v$ is a follower node, then this step is skipped. 
    \item \underline{Next target outcome:} After moving to a node $w$ from $v$, this step chooses the target outcome to commit from $w$. 
\end{enumerate}

The algorithm starts by performing the $\textsc{Strategy}(v_0,p^*,m_0)$ where $m_0$ is an initial memory state denoting the empty memory. 
After finishing the 3 steps of $\textsc{Strategy}(v,p,m)$ for some a given $v,p,m$, the algorithm continues by performing $\textsc{Strategy}(w,p',\memupdate(m,v,w))$ for each $w\in C(v)$ obtained in the Step 2 and the corresponding next target outcome $p'$ obtained in the Step 3.
The downward pass is finished when a leaf node is reached.
We now describe in detail these 3 steps for specific $v,p,m$.

\begin{algorithm}
  \caption{\textbf{Downward Pass}}\label{algp:dow}
  \begin{algorithmic}[1]
    \State $v \leftarrow v_0$ \Comment{$v_0$ is the root node}
    \State $p \leftarrow p^*$ \Comment{$p^*$ is outcome in $S_{v_0}$ with highest utility for leader}
    \State Construct $G_{v_0,p^*}$ \Comment{$G_{v_0,p^*}$ is mixed path graph constructed using $L(v_0,p^*)$} 
    \State $m \leftarrow m_0$ \Comment{$m_0$ is initial memory constructed using $G_{v_0,p^*}$, details are in the proof}
    \State \textsc{Strategy($v_0,p^*,m_0$)}

\phase{Procedure \textsc{Strategy}}    

    \Procedure{Strategy}{$v,p,m$}
        \While{\texttt{$v$ is not leaf node}}
        \State $v',p,',m' \leftarrow \textsc{NewTarget($v,p,m$)}$
        \State \textsc{Strategy($v',p',m'$)}
        \EndWhile
    \EndProcedure

\phase{Memory update function}        

    \Procedure{MemoryUpdate}{$v,p,m$}
        \If{\texttt{$m=R$}} 
            \Return $R$
        \EndIf
        \If{\texttt{$v$ is leader node}}
            \If{\texttt{Leader constructs randomization device}}
                \State \Return $m'$ \Comment{The details of randomization device and $m'$ are in the proof}
            \Else \Return $m$
            \EndIf
        \Else 
            \If{\texttt{$(v,w)$ is in $m$ and follower does not play to $w$}}
                \Return $R$
            \Else \Return $m$
            \EndIf
        \EndIf

    \EndProcedure

\phase{Commitment strategy for the leader}    

    \Procedure{Commitment}{$v,p,m$}
        \If{\texttt{$v$ is follower node}} \Return
        \EndIf
        \If{\texttt{$m=R$}}
            \State Commit to a punishing strategy
            \State \Return $(v',p')$ \Comment{$v',p'$ are the node and outcome suggested by the punishing strategy}
        \Else 
            \If{\texttt{L(v,p) = (1,w,p)}}
                \State Commit to $w$
                Return $(w,p)$
            \EndIf
        \EndIf
    \EndProcedure

\phase{New target node and outcome}    

    \Procedure{NewTarget}{$v,p,m$}
        \If{\texttt{$v$ is leader  node}}
            \State $(w,p) \leftarrow \textsc{Commitment($v,p,m$)}$
            \State \Return $(w,p,\textsc{MemoryUpdate($v,p,m$)})$    
        \Else 
            \If{\texttt{\textsc{MemoryUpdate($v,p,m$)} = R}}
                \State \Return $(v,'p',R)$ \Comment{$v',p'$ are the node and outcome suggested by the punishing strategy}
            \ElsIf{\texttt{$L(v,p) = (1,w,p)$}}
                \State \Return $(w,p,m)$
            \ElsIf{\texttt{$L(v,p) = (\alpha,w,p_1,1-\alpha,w,p_2)$}}
                \State \Return $(w,p_1,m)$ with probability $\alpha$ and \Return $(w,p_2,m)$ with probability $1-\alpha$    
            \EndIf
        \EndIf
    \EndProcedure
  \end{algorithmic}
\end{algorithm}

\noindent \underline{Memory update:}\space 

We shall first define some terminology that we need to construct the memory update function.
Set of memory states $M$ is initialized to an empty set.
Each time the procedure $strategy$ is performed, if a new memory state is created, it is added to $M$. 
Each memory state within the set $M$ is again stored as a set.
Every memory state contains encoded information of two types: 
(1) Move suggestion (m): This type of memory recommends the follower to take a specific action. For instance, if $(v,w)$ is an element of a memory state $m$, then follower at node $v$ is recommended to move to node $w$ on observing memory $m$. 
(2) Red flag (R): This memory informs the leader if the follower has not obeyed the recommendation in the past. There is a unique memory state of this kind. The leader is recommended to use a punishing strategy on observing the memory $R$.

A mixed path $G_{v,p}$ corresponding to a node $v$ and an outcome $p$ is a graph which imitates the recommendations of the label $L(v,p)$.
$G_{v,p}$ can be constructed using $L(v,p)$ as follows.
$G_{v,p}$ is initialized as an empty graph with same set of nodes as the original graph. 
The operation $addedge(w,p')$ adds an edge from $w$ to $w'$ where $L(w,p') = (1,w',p')$ or $L(w,p') = (\alpha,w',p_1,1-\alpha,w',p_2)$.    
Edges are sequentially added in $G_{v,p}$ starting by performing $addedge(v,p)$. 
For any $w,p'$ if $addedge(w,p')$ is performed, then if $L(w,p') = (1,w',p')$ then $addedge(w,p')$ is also performed, and if $L(w,p') = (\alpha,w',p_1,1-\alpha,w',p_2)$ then $addedge(w',p_1)$ and $addedge(w',p_2)$ are also performed.
Edge duplication is not allowed.

A node $w$ is \emph{visited} in the mixed path $G_{v,p}$ if there is edge in $G_{v,p}$ incoming to $w$.
Let $successors(v,p)$ be the set of all nodes that are visited in $G_{v,p}$.
We say that a label $L(w,p')$ is contained in label $L(v,p)$ if the mixed path $G_{w,p'}$ is a subgraph of $G_{v,p}$.
$G_{v,p}$ \emph{converges} (or \emph{diverges}) at node $w$ if $w$ has more than 2 incoming (or outgoing) edges in $G_{v,p}$.
Two mixed paths $G_{v,p}$ and $G_{v',p'}$ \emph{converge} (or \emph{diverge}) at node $w$ if $w$ has at least one incoming (or outgoing) edge in both the $G_{v,p}$ and $G_{v',p'}$. 
The memory state at the root node denoted by $m_0$ is constructed using the mixed path $G_{v_0,p^*}$. 
$m_0$ is initialized as an empty set. 
For every follower node $v$ at which $G_{v_0,p^*}$ does not diverge, add $(v,w)$ to $m_0$ where $(v,w)$ is the unique edge outgoing from $v$ in $G_{v_0,p^*}$.
The memory state $m_0$ is added to $M$.

Firstly $\memupdate(R,v,w) = R$ for any node $v$ and $w \in C(v)$, that is, if the follower does not obey the recommendations at some node, this Red flag memory does not change till the end.
Consider a follower node $v$. 
When the memory $m$ contains an element $(v,w)$ for some $w \in C(v)$, if follower moves $w$, then $\memupdate(m,v,w') = m$ and if follower moves to $w' \in C(v)$ other than $w$, then $\memupdate(m,v,w') = R$.
When $m$ does not contain an element $(v,w)$ for any $w \in C(v)$, then $\memupdate(m,v,w) = m$ for any $w \in C(v)$.

Construction of memory update function for the leader nodes is bit more complicated, as it depends on the mixed labels of past follower nodes. 
Consider a leader node $v$, then $\memupdate(m,v,w)$ returns a new memory state only in the following scenario.
Let there be a follower node $v'$ visited in $G_{v_0,p^*}$ with mixed label $L(v',p) = (\alpha,w',p_1,1-\alpha,w',p_2)$ such that 
as the mixed paths $G_{w,p_1}$ and $G_{w,p_2}$ traverse down the graph, diverge for the first time at the leader node $v$, say to $w_1$ and $w_2$ respectively where $w_1,w_2 \in C(v)$.
If there is a follower node $w'$ in $successors(w_1,p_1) \cap successors(w_2,p_2)$ and the mixed paths $G_{w,p_1}$ and $G_{w,p_2}$ diverge again at $w'$, say $G_{w,p_1}$ goes to $w'_1$ and $G_{w,p_2}$ goes to $w'_2$ (where $w'_1,w'_2 \in C(w')$), then algorithm creates two new memory states $m_{w_1}$ and $m_{w_2}$.  
Initialize $m_{w_1} = m_{w_2} = m$.
For every such $w'$ satisfying the above, add $(w',w'_1)$ to $m_{w_1}$ and $(w',w'_2)$ to  $m_{w_2}$.
Finally $\memupdate(m,v,w_1) = m_{w_1}$ and $\memupdate(m,v,w_2) = m_{w_2}$, $m_{w_1}$ and $m_{w_2}$ are added to $M$. For every other $w"\in C(v)$ other than $w_1,w_2$, $\memupdate(m,v,w") = m$
If either of the above conditions are not met, then $\memupdate(m,v,w") = m$ for each $w" \in C(v)$.

\vspace{2mm}
\noindent \underline{Commitment of the leader:}\space If the memory $m_v = R$ at any leader node $v$, then the algorithm commits to a punishing strategy for the leader at node $v$ for the reminder of the game. 
If $m_v \neq R$, then the commitment depends on the label $L(v,p)$.
If $L(v,p) = (1,w,p)$ for some $w \in C(v)$, then the leader should simply commit to $w$ with probability 1.

Note that the label $L(v,p)$ for any leader node $v$ and for any outcome $p$ is always a pure label, so the leader is committed to play pure action conditional on the fact that the label $L(v,p)$ is reached. 
However when the labels $L(v,p) = (1,w,p)$ are $L(v,p') = (1,w',p')$ are both reached reached with positive probabilities $\alpha$ and $1-\alpha$ respectively for $w'$ different from $w$, then the leader should be committed to a $w$ with probability $\alpha$ and to $w'$ with probability $1-\alpha$.  

\vspace{2mm}
\noindent \underline{Next target outcome:}\space 
Consider a node $v$ with target outcome $p$ and memory $m_v$. If $L(v,p) = (1,w,p)$, then in the next step $\textsc{Strategy}(w,p,\memupdate(m,v,w))$ is performed. 
If $L(v,p) = (\alpha,w,p_1,1-\alpha,w,p_2)$, then $\textsc{Strategy}(v,p_1,m_v)$ and $\textsc{Strategy}(v,p_2,m_v)$ are performed, with probabilities $\alpha$ and $1-\alpha$ respectively.
\qed

\vspace{2mm}

\noindent{\bf Runtime Analysis}\space 

We start by bounding the size of $S_v$ for each $v$.
Each outcome in $S_v$ is either picked up from $S_w$ for some $w\in C(v)$, or it is newly created while constructing $S_v$.
The new outcomes are created in two ways. 
Firstly an outcome can be generated directly a leaf node, number of such outcomes is bounded by the number of leaves, that is $|L|$.
Secondly an outcome can be created by a mixture of two outcomes in $S_w$ for some $w \in C(v)$.
Since at most one such outcome can be created for any $w \in C(v)$, the number of outcomes of second type can be bounded by $|C(v)|$. 
So for any node $v$, $|S_v|$ is bounded by $|V| + |L|$ or simply by $2|V|$.

During the upward pass, to construct $S_v$ for node the algorithm needs to deal each pair of outcomes in $S_w$ for each $w \in S_v$.
Thus for each pair of $v$ and $w \in C(v)$ we need to compute the value $i(w)$ (which takes $O(|C(v)||V|)$ operations) and for each pair of outcomes we need to decide whether a new outcome should be created (which takes $O(|S_w|^2)$ operations). 
Hence for each pair of $v,w$, we need $O(|V|^2)$ operations. 
So the upward pass needs $O(|V|^4)$ operations.
Choosing the outcome optimal for the leader takes $ |S_{v_0}| = O(|V|)$ operations. 

During the downward pass, for any $v$ and $p$, mixed path $G_{v,p}$ can be constructed in $O(|V|^2)$ operations, since each node and edge in the graph is traversed at most once. 
Construction of $m_0$ requires $O(|V|^2)$ operations. 
While performing the procedure $\textsc{Strategy}(v,p,m)$ if the $v$ is a follower node, then $\textsc{MemoryUpdate}(v,p,m)$ is performs $O(1)$ operations. 
When $v$ is a leader node, whether a randomization device should be created can be checked in $O(|V|^3)$ operations. 
Since at most constant number of memory states are created corresponding to each node, the size of set of memory states is $O(|V|)$. 
Commitment at each leader node takes $O(1)$ operations and next target node and outcome can also be determined in $O(1)$ operations.
Since $\textsc{Strategy}$ is performed at most $O(|V|)$ times, the downward pass takes $O(|V|^5)$ operations.

\begin{examp} \rm

\begin{figure}
\centering
\begin{tikzpicture}
		[node distance=0.4cm and 0.8cm, label distance=-8pt]
	  \tikzstyle{leader}=[circle,draw,inner sep=1.5]
   	\tikzstyle{leaf}=[]
 	  \tikzstyle{follower}=[rectangle,draw,inner sep=2]
  	\tikzstyle{nature}=[diamond,draw,inner sep=1.5]
  	\node(a)[follower]{$v_0$};
  	\node(b)[leader, right=of a]{$v_2$};
  	\node(d)[follower, above right=of b]{$v_3$};
  	\node(e)[follower, below right=of b]{$v_4$};
  	\node(g)[follower, above right=of e]{$v_5$};
  	\node(i)[leader, above right=of g]{$v_6$};
  	\node(j)[leader, below right=of g]{$v_7$};
  	\node(l)[leaf, above=of k]{$(3, 4) \hspace{2mm} \bf{l_1}$};
  	 \node(k)[leaf, right=of i]{$(4, 0) \hspace{2mm} \bf{l_2}$};
  	\node(m)[leaf, right=of j]{$(2, 0) \hspace{2mm} \bf{l_3}$};
  	\node(n)[leaf, below=of m]{$(1, 6) \hspace{2mm} \bf{l_4}$};
  	 \node(c)[leaf, below=of a]{$(0,5) \hspace{2mm} \bf{l_5}$};
  	\draw[->] 
		  (b) edge node[label=above left:{$a^2_u$}]{} (d)
		  (b) edge node[label=below left:{$a^2_d$}]{} (e)
          (d) edge (g) 
          (e) edge (g)
		  (g) edge node[label=above left:{$a^5_u$}]{} (i)
		  (g) edge node[label=below left:{$a^5_d$}]{} (j)
		  (i) -- node[midway,sloped,below]{$a^5_d$} (k)
		  (i) edge node[label=above left:{$a^6_u$}]{} (l)
		  (j) edge node[label=below left:{$a^7_d$}]{} (n);
		  \draw[->] (j) -- node[midway,sloped,above]{$a^7_u$} (m);
		  \draw[->] (a) -- (b) node[midway,sloped,above]{$a^1_r$};
		  \draw[->] (a) -- (c) node[midway,left]{$a^1_d$};
\end{tikzpicture}
\caption{An example of a game on a DAG. The leader plays in circular nodes,
the follower plays in square nodes.}
\label{dagexampleagain}
\end{figure}
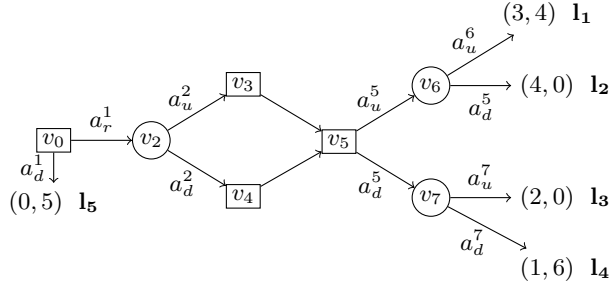

We will revisit the example from Section \ref{Sec:Memory} to demonstrate how our algorithm works.
In the upward pass, the algorithm computes the sets $S_v$ for each node $v$ in a bottom up manner. 
The algorithm starts with leaf nodes: $S_{l_i} = \{l_i\}$ and $L(l_i,l_i) = ()$  for $i=1,2,3,4,5$.
At the leader node $v_6$, leader can commit to both $l_1$ and $l_2$, so $S_{v_6} = \{l_1,l_2\}$ and $L(v_6,l_1) = (1,l_1,l_1)$ and $L(v_6,l_2) = (1,l_2,l_2)$. 
Similarly, $S_{v_7} = \{l_3,l_4\}$.
At the follower node $v_5$, the algorithm checks all the outcomes in $S_{v_6}$ and $S_{v_7}$. 
Firstly to check whether $l_1$ should be added to $S_{v_5}$, the algorithm computes the quantities
$u_2(l_1) = 4$ and
$i(v_6) = \max_{w' \in C(v_5),w' \neq v_6}\min_{p'\in S_{w'}} u_2(p') = u_2(l_3) = 0$.
Since $u_2(l_1) \ge i(v_6)$, $l_1$ is added to $S_{v_5}$. 
As $u_2(l_2) = 0 \ge i(v_6)$, $l_2$ is also added to $S_{v_5}$.
With similar reasoning, $l_3$ and $l_4$ are also added to $S_{v_5}$. 
Hence $S_{v_5} = \{l_1,l_2,l_3,l_4\}$. 
Labels at $v_5$ are defined as $L(v_5,l_1) = (1,v_6,l_1)$, $L(v_5,l_2) = (1,v_6,l_2)$, $L(v_5,l_3) = (1,v_7,l_3)$ and $L(v_5,l_4) = (1,v_7,l_4)$.
Continuing upwards, the sets $S_{v_3} = S_{v_4} = S_{v_2} = \{l_1,l_2,l_3,l_4\}$ are constructed. 

Now consider the follower node $v_0$, which is also a root node.
Since the follower can go to leaf node $v_5$ and obtain utility 5, we have $i(v_2) = 5$. 
Since $u_2(l_1),u_2(l_2),u_2(l_3) < i(v_2)$, the outcomes $l_1,l_2,l_3$ are not added to $S_{v_0}$, and since $u_2(l_4) \ge i(v_2)$, $l_4$ is added to $S_{v_0}$. 
Also $l_5$ is added to $S_{v_0}$ since $5 = u_2(l_1) \ge i(l_5) = u_2(l_2) = 0$.
Finally, all the pairs of leaves are considered to check if any mixed outcome should be added to $S_{v_0}$. 
It can be seen that since $u_2(\frac{1}{2}(l_1)+\frac{1}{2}(l_4)) = 5 = i(v_2)$, outcome $\frac{1}{2}(l_1)+\frac{1}{2}(l_4)$ is added to $S_{v_0}$.
Similarly, the outcomes $\frac{1}{6}(l_2)+\frac{5}{6}(l_4)$ and $\frac{1}{6}(l_3)+\frac{5}{6}(l_4)$ are added to $S_{v_0}$.
Hence, $S_{v_0} = \{l_4,l_5,\frac{1}{2}(l_1)+\frac{1}{2}(l_4),\frac{1}{6}(l_2)+\frac{5}{6}(l_4),\frac{1}{6}(l_3)+\frac{5}{6}(l_4)\}$. 
Finally the outcome $p^* = \frac{1}{2}(l_1)+\frac{1}{2}(l_4)$ is chosen since the leader obtains largest expected utility for this outcome among all the outcomes in the convex hull of $S_{v_0}$. 
Since $p^*$ is newly created by mixing two different outcomes, the label would be a mixed label and defined as $L(v_0,p^*) = (\frac{1}{2},v_2,l_1,\frac{1}{2},v_2,l_4)$.
Labels for the other outcomes are defined similarly.
This completes the upward pass.

The downward pass begins with the construction of the mixed path $G_{v_0,p^*}$ where $p^* = \frac{1}{2}(l_1)+\frac{1}{2}(l_4)$.
Using the labels defined in the upward pass, the edges $(v_0,v_2)$, $(v_2,v_3)$, $(v_2,v_4)$, $(v_3,v_5)$, $(v_4,v_5)$, $(v_5,v_6)$, $(v_5,v_7)$, $(v_6,l_1)$, $(v_6,l_2)$, $(v_6,l_3)$, $(v_6,l_4)$ are added.
Initial memory $m_0 = \{(v_0,v_2),(v3,v_5),(v_4,v_5)\}$ is also constructed, since there is a unique outgoing edge in $G_{v_0,p^*}$ from the follower nodes $v_0,v_3$ and $v_4$.

Procedure \textsc{Strategy($v_0,p^*,m_0$)} is performed. 
As $(v_0,v_2)\in m_0$, if follower moves to $l_5$, memory is updated to $R$. 
If follower moves to $v_2$, then memory is not changed and since $L(v_0,p^*) = (\frac{1}{2},v_2,l_1,\frac{1}{2},v_2,l_4)$, procedures \textsc{Strategy($v_2,l_1,m_0$)} and \textsc{Strategy($v_2,l_4,m_0$)} are performed with probability $\frac{1}{2}$ each.
While performing \textsc{Strategy($v_2,l_1,m_0$)}, since $L(v_2,l_1) = (1,v_3,l_1)$, leader commits to $v_3$ with probability 1 and a new memory state $m_3$ is created and defined as $m_3 = m_0 \cup \{v_5,v_6\}$. This acts as a suggestion to follower to move to $v_6$ after node $v_5$.  
Whereas after performing \textsc{Strategy($v_2,l_4,m_0$)}, since $L(v_2,l_1) = (1,v_4,l_4)$, leader commits to $v_4$ with probability 1 and a new memory state $m_4$ is created and defined as $m_4 = m_0 \cup \{v_5,v_7\}$. 
The node $v_5$ is reached in two ways, with probability $\frac{1}{2}$ via $v_3$ and has memory state $m_3$ and with probability $\frac{1}{2}$ via $v_4$ and has memory state $m_4$.
On observing $m_3$ if follower moves to $v_7$, memory state is changed to $R$ since $(v_5,v_6) \in m_3$. In this case the leader plays a punishing strategy, that is moves to $l_3$ so that the follower gets 0. 
On observing $m_3$ if follower moves to $v_6$, memory state is not changed and \textsc{Strategy($v_6,l_1,m_3$)} is performed.
Similarly on observing memory $m_4$ if follower moves to $v_6$. leader plays punishing strategy and if follower moves to $m_7$, \textsc{Strategy($v_7,l_4,m_4$)} is performed.

\end{examp}

\vspace{-0.3cm}
\section{Algorithm for Games Without Chance Nodes on Directed Graphs}\label{sec:algoDG}

In this section, we consider the games on DGs without chance nodes. 
When there are directed cycles, the game has infinite horizon.
The players still receive utility only after a leaf node is reached. 
Due to the existence of directed cycles, players can move indefinitely without ever reaching a leaf node. 
Since we assume that infinite play gives zero payoff to both the players, it might be beneficial for players to move along a cycle indefinitely, if the terminal payoffs are negative. 

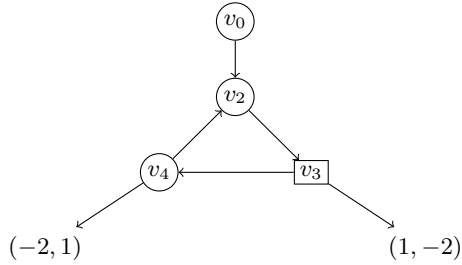
\begin{figure}
\centering
\begin{tikzpicture}
		[node distance=0.4cm and 0.8cm, label distance=-8pt]
	  \tikzstyle{leader}=[circle,draw,inner sep=1.5]
   	\tikzstyle{leaf}=[]
 	  \tikzstyle{follower}=[rectangle,draw,inner sep=2]
  	\tikzstyle{nature}=[diamond,draw,inner sep=1.5]
  	
  	\node(t) at (0,3) [leader]{$v_0$};
  	\node(s) at (0,2) [leader]{$v_2$};
  	\node(u) at (1,1) [follower]{$v_3$};
  	\node(q) at (-1,1) [leader]{$v_4$};

  	\node(x) at (2.5,0) {$(1,-2)$};
  	\node(a) at (-2.5,0){$(-2,1)$};

  	\draw[->] 

  	      (t) edge (s)
  	      (s) edge (u)
  	      (u) edge (q)
  	      (q) edge (s)
  	      (u) edge (x)
  	      (q) edge (a)
		  ;
\end{tikzpicture}
\caption{Here in the unique SSE game is played indefinitely and both players receive zero. }
\label{fig:infinite}
\end{figure}

\begin{examp}\rm \label{ex:infinite}
The Fig. \ref{fig:infinite} shows an example of a game containing directed cycle where some of the entries in terminal payoff are negative. In this game, there is a unique SSE in which the players move indefinitely in the cycle with nodes $v_2,v_3$ and $v_4$ and none of the leaf is ever reached.
\hfill $\diamond$
\end{examp}

To avoid scenarios like in Example  \ref{ex:infinite}, we assume in the rest of the sections that all the terminal payoffs for both the players are non-negative.

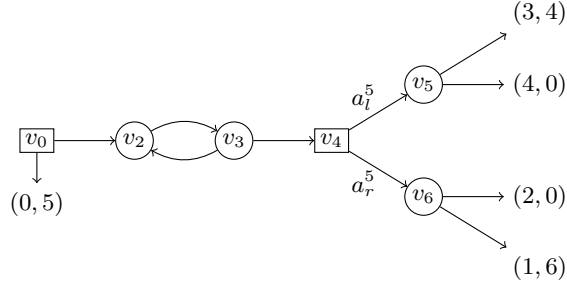
\begin{figure}
\centering
\begin{tikzpicture}
		[node distance=0.4cm and 0.8cm, label distance=-8pt]
	  \tikzstyle{leader}=[circle,draw,inner sep=1.5]
   	\tikzstyle{leaf}=[]
 	  \tikzstyle{follower}=[rectangle,draw,inner sep=2]
  	\tikzstyle{nature}=[diamond,draw,inner sep=1.5]
  	\node(a)[follower]{$v_0$};
  	\node(b)[leader, right=of a]{$v_2$};
  	\node(o)[leader, right=of b]{$v_3$};
  	\node(c)[leaf, below=of a]{$(0,5)$};
  	\node(g)[follower, right=of o]{$v_4$};
  	\node(i)[leader, above right=of g]{$v_5$};
  	\node(j)[leader, below right=of g]{$v_6$};
  	\node(k)[leaf, right=of i]{$(4, 0)$};
  	\node(l)[leaf, above=of k]{$(3, 4)$};
  	\node(m)[leaf, right=of j]{$(2, 0)$};
  	\node(n)[leaf, below=of m]{$(1, 6)$};
  	\draw[->] 
  	      (a) edge (b) 
  	      (a) edge (c)
  	      (b) edge [bend left] node [left] {} (o)
  	      (o) edge [bend left] node [left] {} (b)
  	      (o) edge (g)
		  (g) edge node[label=above left:{$a^5_l$}]{} (i)
		  (g) edge node[label=below left:{$a^5_r$}]{} (j)
		  (i) edge (k) 
		  (i) edge (l) 
		  (j) edge (m) 
		  (j) edge (n);
\end{tikzpicture}
\caption{An example of a game on a DG. In SSE, the nodes $v_2$ and $v_3$ are visited twice.}
\label{dgexample1}
\end{figure}

\begin{examp}\rm \label{ex:DG}
Fig. \ref{dgexample1} describes an example of a game which is small variant of the one described by Fig.\ref{dagexample}. 
In this game any optimal commitment strategy of the leader ensures that the states $v_2$ and $v_3$ are visited more than once with positive probability. 
Like in Fig.\ref{dagexample}, the leader can obtain the payoff 2 if after state $v_4$, the states $v_5$ and $v_6$ are reached with equal probabilities. 
This can be achieved by creating a randomization device. 
For instance suppose, on the first visit at $v_3$ the leader goes to $v_4$ and $v_2$ with equal probabilities, and on second visit at $v_3$, the leader goes to $v_4$ with probability 1. 
Now, the leader can commit at nodes $v_5$ and $v_6$ so as to `recommend' the follower that at node $v_4$ he must go to $v_5$ if in his observed history $v_3$ was visited only once and to $v_6$ if in his observed history $v_3$ was visited twice. 
Since the leader also observes the same history, he can punish the follower by going to $(4,0)$ after $v_6$ and going to $(2,0)$ after $v_7$ if the follower does not follow the recommendation, thus giving the follower a strict incentive to follow the leader's recommendation.
\hfill $\diamond$
\end{examp}

\begin{figure}
\centering
\begin{tikzpicture}
		[node distance=0.4cm and 0.8cm, label distance=-8pt]
	  \tikzstyle{leader}=[circle,draw,inner sep=1.5]
   	\tikzstyle{leaf}=[]
 	  \tikzstyle{follower}=[rectangle,draw,inner sep=2]
  	\tikzstyle{nature}=[diamond,draw,inner sep=1.5]
  	
  	\node(t) at (0,3) [follower]{$v_0$};
  	\node(v) at (1.5,3) [leaf]{$(0,5.5)$};
  	\node(s) at (0,2) [leader]{$v_2$};
  	\node(u) at (1,1) [leader]{$v_3$};
  	\node(q) at (-1,1) [leader]{$v_4$};

  	\node(x) at (2.5,0) [follower]{$v_5$};
  	\node(y)[leader, above right=of x]{$v_6$};
  	\node(z)[leader, below right=of x]{$v_7$};
  	\node(w)[leaf, right=of y]{$(4, 0)$};
  	\node(l)[leaf, above=of w]{$(3, 4)$};
  	\node(p)[leaf, right=of z]{$(2, 0)$};
  	\node(r)[leaf, below=of p]{$(1, 6)$};
  	
  	\node(a) at (-2.5,0) [follower]{$v_8$};
  	\node(b)[leader, above left=of a]{$v_9$};
  	\node(c)[leader, below left=of a]{$v_{10}$};
  	\node(d)[leaf, left=of b]{$(3, 1)$};
  	\node(e)[leaf, above=of d]{$(2, 5)$};
  	\node(f)[leaf, left=of c]{$(1, 1)$};
  	\node(g)[leaf, below=of f]{$(0, 7)$};

  	\draw[->] 
  	      (t) edge (v)
  	      (t) edge (s)
  	      (s) edge (u)
  	      (u) edge (q)
  	      (q) edge (s)
  	      (u) edge (x)
  	      (q) edge (a)
  	
		  (x) edge (y)
		  (x) edge (z)
		  (y) edge (w)
		  (y) edge (l)
		  (z) edge (p)
		  (z) edge (r)
		  
		  (a) edge (b)
		  (a) edge (c)
		  (b) edge (d)
		  (b) edge (e)
		  (c) edge (f)
		  (c) edge (g)
		  
		  ;
\end{tikzpicture}
\caption{Another example of a game on a DG. In SSE, the nodes $v_2$, $v_3$ and $v_4$ are visited twice. }
\label{dgexample2}
\end{figure}
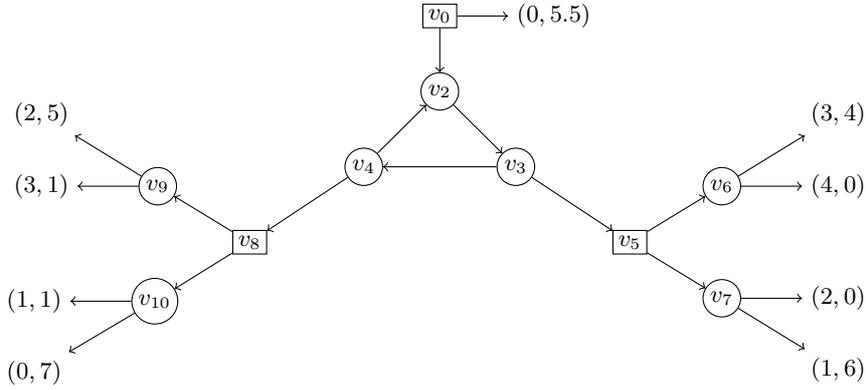

\begin{examp}\rm \label{ex:DG2}
The Fig. \ref{dgexample2} shows an example similar to the Example \ref{ex:DG}. In any SSE in which the outcomes (3,4), (1,6), (2,5), (0,7) are achieved with equal probabilities, hence the expected payoff vector is (1.5,5.5). 
We can describe the strategy of the leader in one such SSE as follows.
On the first visit at $v_3$, the leader plays $v_5$ with probability $\frac{1}{4}$ and $v_4$ with probability $\frac{3}{4}$. 
On the first visit at $v_4$, he plays $v_8$ with probability $\frac{1}{3}$ and $v_4$ with probability $\frac{2}{3}$.
On the second visit at $v_3$, the leader plays $v_5$ with probability $\frac{1}{2}$ and $v_4$ with probability $\frac{1}{2}$. 
On the first visit at $v_4$, he plays $v_8$ with probability 1.
At $v_5$, the follower is incentivized to play $v_6$ if in the history visited $v_3$ only once, and to play $v_7$ if in the history visited $v_3$ twice.
Similarly at $v_8$, the follower is incentivized to play $v_9$ if in the history visited $v_4$ only once, and to play $v_{10}$ if in the history visited $v_4$ twice.
This incentivization is achieved by using punishing strategies.
\hfill $\diamond$
\end{examp}

When the graph has with directed cycles, there are two main issues which make the computation of a SSE more difficult compared to the case with DAGs. 
Firstly, due to the existence of directed cycles there does not exist a reverse topological order of the nodes, and so, straightforward dynamic programming techniques do not work. 
Secondly, as discussed in examples \ref{ex:DG} and \ref{ex:DG2}, there are games in which in any SSE, one or more nodes are visited more than once with positive probability. 
The outcomes to which leader can commit starting at a node by visiting certain nodes more than once, might not be included in the commitment set of that node which is obtained by a single iteration of the upward pass identical to that of Section \ref{sec:algoDAG}. 
So unlike in the case of of DAGs, a single iteration of an upward pass may not sufficient to accurately construct the commitment sets.

The first issue can be solved relatively easily by decomposing the graph into strongly connected components (SCCs) and dealing with the SCCs as a whole by fixing a reverse topological order on SCCs.
In order to circumvent the second issue, we show (Proposition \ref{prop:twotimes}) that there exists a SSE in which each node is visited at most 2 times. 
So, we can modify the algorithm from section \ref{sec:algoDAG} such that upward pass is performed twice followed by the exact same downward pass. 
This allows us to use Theorem \ref{thm:DAG} and Proposition \ref{prop:twotimes} to construct a polynomial time algorithm that computes a SSE where the players use strategy profile with memory. 

The main intuition for Proposition \ref{prop:twotimes} is as follows. 
Since the players do not receive any payoff before a leaf is reached, their final payoff only depends on the terminal node, and not on the exact path taken.
However by following certain paths with positive probability, the leader the leader can incentivize the actions of the follower so that the desired leaf nodes are reached with positive probability. 
On certain paths, a node may need to be visited more than once, so that two different possible histories at that node.
Leader can use this fact to create a randomization device. 
On the other hand, visiting a node more than two times does not help the leader any way to obtain better outcomes.   

\begin{prop}\label{prop:twotimes}\rm
In any game on DG without chance nodes, there exists a SSE in which each node is visited at most two times, in any path with positive probability. 
\end{prop}

\begin{proof}

We first partition the graph $G$ into strongly connected components (SCC). 
Note that, every leaf node in $G$ forms their own singleton SCC.
We construct a directed graph $G_{SCC}$ as follows.
For each SCC $C$ in $G$, there is a corresponding vertex $C$ in $G_{SCC}$.    
$C'$ is a child of $C$ (and $C$ is a parent of $C'$) in $G_{SCC}$ if there exists an edge in $G$ from a vertex in $C$ to a vertex in $C'$ (we do not allow self loops and multiple edges in $G_{SCC}$). 
$C'$ is a successor of $C$ in $G_{SCC}$ if there exists a sequence $C=C_1,\ldots,C_k=C'$ such that $C_{i+1}$ is a child of $C_i$ for $i = 1,\ldots,k-1$ in .
By construction, $G_{SCC}$ is a DAG.
Fix a reverse topological order on SCCs.

Let $(\sigma_1,\sigma_2)$ be a SSE and we assume that $\sigma_2$ is a pure strategy (this assumptions is without loss of generality due to Lemma \ref{lemm:follower-pure}). 
We will modify $\sigma_1$ to construct a new strategy $\hat{\sigma_1}$ such that $(\hat{\sigma_1},\sigma_2)$ is also a SSE and each node is visited at most twice. 
We initialize $\hat{\sigma_1}(h) = \sigma_1(h)$ for each $h \in H$, in each step we will modify $\hat{\sigma_1}$ at a single SCC, according the reverse topological order.

Every time we modify $\hat{\sigma_1}$ at a SCC, we show that $\sigma_2$ is still the best response to the modified $\hat{\sigma_1}$, and hence modified $(\hat{\sigma_1},\sigma_2)$ continues to be a SSE. 

Now we will explain how to modify $\hat{\sigma_1}$ at a fixed SCC $C$. 
To avoid confusion, we use notations $\sigma'_1$ and $\sigma''_1$ respectively, to denote the strategy $\hat{\sigma_1}$ just before its modification in $C$ and just after its modification in $C$.

Note that once a play exits from a SCC $C$, it can never enter $C$ again.
Let $H_C$ denote the set of histories which enter $C$ for the first time,  
in other words, the set of histories of the form $v_0,\ldots,v_k$ such that $v_k$ is in $C$ and $v_i$ for $i<k$ is not in $C$.
For $h\in H_C$, let $\Pi^h$ denote the set of plays which are continuations of $h$.
Let $H_C^h$ denote the set of histories which are continuations of $h$ and end at a vertex in $C$, in other words, the set of continuations of $h$ of the form $v_0,\ldots,v_k$ such that $v_k$ is in $C$.
Similarly let $E_C^h$ denote the set of histories, which are continuations of $h$ and exit $C$ for the first time, in other words, the set of continuations of $h$ of the form $v_0,\ldots,v_k$ such that $v_{k-1}$ is in $C$ and $v_k$ is not in $C$.

Fix a history $h \in H_C$ and let $v_h \in C$ be the final node of $h$.  
We now will define the behaviour of $\sigma''_1$ at every history in $H_C^h$, so that in any play in $\Pi^h$ with positive support in $(\sigma''_1,\sigma_2)$, each node in $C$ is visited at most twice. 
If every $h'\in H_C^h$ with positive support in $(\sigma'_1,\sigma_2)$ visits each node in $C$ at most twice, then we let $\sigma''_1(h') = \sigma'_1(h')$ for each $h' \in H_C^h$. 
Now assume that there is a history in $H_C^h$ with positive support in $(\sigma'_1,\sigma_2)$ which visits some node in $C$ more than twice.

We will now provide some intuition behind the construction of $(\sigma''_1,\sigma_2)$. We will construct $(\sigma''_1,\sigma_2)$ in a way that each play in $\Pi^h$ exits $C$ via at most two histories. This implies that the leader randomizes in at most one node in $C$ (note that $\sigma_2$ is a pure strategy). So, if there is a path $\pi \in \Pi^h$ such that a node $v$ in $C$ is visited at least 3 times, no randomization takes place between two consecutive visits of $v$. So, we can modify $\sigma''_1$ so that this part is short circuited to ensure that $v$ is visited at most twice. We formalize the construction as follows.

Each play in $\Pi^h$ is a continuation of exactly one history in $E_C^h$. 
We will now show that it is possible to construct $\sigma''_1$ in a way that there are at most two histories in $E_C^h$, such that any play in $\Pi^h$ with positive support in the $(\sigma''_1,\sigma_2)$, is continuation of one of them.
Let $h_1,\ldots,h_k$ be the histories in $E_C^h$ with positive support in the $(\sigma'_1,\sigma_2)$, assume that $h_i$ is realized with probability $p_i(>0)$. Note that $p_1 + \ldots + p_k = p$.
Due to Lemma \ref{lemm:onlytwo}, there exist two histories $h_1$ and $h_2$, and probability $\hat{p}$ such that 

\vspace{-5mm}
\begin{equation}\label{eq:onlytwo}
\hat{p} \cdot u(\sigma'_1,\sigma_2)(h_1) + (p-\hat{p}) \cdot u(\sigma'_1,\sigma_2)(h_2) \ge p_1 \cdot u(\sigma'_1,\sigma_2)(h_1) + \ldots + p_k \cdot u(\sigma'_1,\sigma_2)(h_k)
\end{equation}

As the histories $h_1$ and $h_2$ are realized with positive probabilities in SSE $(\sigma'_1,\sigma_2)$, due to equation \ref{eq:onlytwo} the leader can commit to strategy (which will be $\sigma''_1$) which achieves the histories $h_1$ and $h_2$ with probabilities $\hat{p}$ and $p-\hat{p}$ respectively. 
Since $\sigma'_1$ and $\sigma''_1$ are identical at all histories in $E_C^h$, we have $u(\sigma''_1,\sigma_2)(h_1) = u(\sigma'_1,\sigma_2)(h_1)$ and $u(\sigma''_1,\sigma_2)(h_2) = u(\sigma'_1,\sigma_2)(h_2)$.
Hence, $\sigma_2$ is best response to $\sigma_2$.
Since $(\sigma''_1,\sigma_2)$ exits $C$ via at most two histories.
\qed
\end{proof} 

\begin{theorem}\label{thm:DG}\rm
For sequential games on DGs without chance nodes,there exists a polynomial time algorithm which computes such a SSE with strategy profile with memory. 
\end{theorem}


\section{Algorithm for Games with Chance Nodes on DAGs}\label{sec:algochanceDG}

In this section, we consider the sequential games which contain chance nodes. 
We restrict the game graphs to DAGs. 
In presence of chance nodes, \cite{letchford2010} show that it is NP-Hard to compute SSE even when the graph is a tree. 
They show this by reducing an arbitrary instance of KNAPSACK problem to an sequential game (on trees with chance nodes). 
Since the reduction restricted to the games on trees, the same reduction works to show the NP-Hardness to compute SSE with memory on games on trees or DAGS. 

\begin{prop}
In sequential games with chance nodes, it is NP hard to compute a SSE with memory. 
\end{prop}

With no hope of finding polynomial time algorithm to find an 
exact SSE, we focus our attention to compute an approximation version of SSE, which gives the leader a payoff arbitrarily close to his SSE payoff. 
To simplify the arguments we transform the graph from a DAG to a binary DAG (the out degree of each node is at most 2). 
This transformation is without loss of generality, since the strategic structure of the game remains same. 
The transformation is done by sequentially adding auxiliary nodes for each node which has more than 2 children. 
These auxiliary nodes ensure that the out degree of each node is at most 2, while maintaining the same strategic structure of the game.
An example of a transformation step in which auxiliary nodes are added when a leader node has 4 children is illustrated in Fig. \ref{binaryDAG}. 
For a node with $k$ children, we require at most $k-1$ auxiliary nodes. 
Hence, the new graph has $O(n^2)$ auxiliary nodes, and the height of the new graph is at most $\log(n)$ more than the original graph. 

\begin{figure}
\centering
\begin{tikzpicture}
		[node distance=0.4cm and 0.8cm, label distance=-8pt]
	\tikzstyle{leader}=[circle,draw,inner sep=1.5]
   	\tikzstyle{leaf}=[]
 	\tikzstyle{follower}=[rectangle,draw,inner sep=2]
  	\tikzstyle{nature}=[diamond,draw,inner sep=1.5]
  	\node(a) at (-3,0.75) [leader]{$v$};

  	\node(b) at (-4.5,-0.75) [follower]{$w_1$};
  	\node(c) at (-3.5,-0.75)[follower]{$w_2$};
  	\node(d) at (-2.5,-0.75) [follower]{$w_3$};
  	\node(e) at (-1.5,-0.75) [follower]{$w_4$};
  	
  	\node(f) at (3,0.75) [leader]{$v$};
  	\node(g) at (2,0.25) [leader]{$a_1$};
  	\node(h) at (4,0.25) [leader]{$a_2$};

    \node(i) at (1.5,-0.75) [follower]{$w_1$};
    \node(j) at (2.5,-0.75) [follower]{$w_2$};
    \node(k) at (3.5,-0.75) [follower]{$w_3$};
    \node(l) at (4.5,-0.75) [follower]{$w_4$};
  	
  	\draw[->] 
  	      (a) edge (b) 
  	      (a) edge (c)
  	      (a) edge (d) 
  	      (a) edge (e)
  	      
  	      (f) edge (g)
  	      (f) edge (h)
  	      
  	      (g) edge (i)
  	      (g) edge (j)
  	      (h) edge (k)
  	      (h) edge (l)
  	      ;
\end{tikzpicture}
\caption{A general DAG can be replaced by a new binary DAG obtained by adding auxiliary nodes, which corresponds to the same game.}
\label{binaryDAG}
\end{figure}
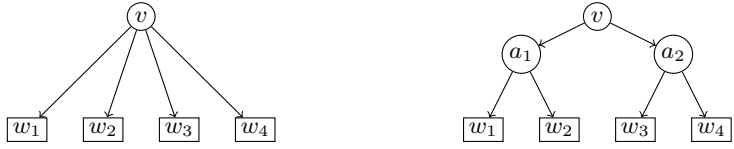


At each node $v$, algorithm constructs the value $A_v[k]$ for each $k \in \{U_1,U_1+\delta,\ldots,U_2\}$  which represents the maximum utility that the follower can obtain starting at $v$, conditional on the fact that the leader can commit to a strategy which with a best response from the follower gives the leader utility at least $k$ starting at $v$. 

\underline{Leaf node:} If $v$ is a leaf where leader obtains utility $u_l$ and follower obtains utility $u_f$, then 
\begin{equation}
A_v[k] =
\begin{cases}
u_f, & \text{if }  k \leq u_l, \\
-\infty, & \text{if }  k > u_l
\end{cases}
\end{equation}

\underline{Chance node:} If a chance node $v$ has only one child $w$, then $A_v[k] = A_w[k]$ for each $k$. 
Now assume that $v$ has two children $L$ and $R$.
If the chance node plays $L$ with probability $p$ and $R$ with probability $1-p$, the payoff $k$ can be guaranteed for the leader by achieving payoff $i$ at the subgame rooted at $L$ and payoff $j$ at subgame rooted at $R$, such that the inequality $pi+(1-p)j \ge k$ is satisfied. 
Hence, the maximum utility for the follower with guarantee that the leader gets at least $k$ can be calculated as the following maximization problem.

\begin{equation}
A_v[k] = \max_{i,j}\set{p A_L[i] + (1 - p) A_R[j] \hspace{2mm} \vert \hspace{2mm} p i + (1 - p) j \geq k}.
\label{eq:chancenode}
\end{equation}
\underline{Follower node:} If a follower node $v$ has only one child $w$, then $A_v[k] = A_w[k]$ for each $k$. Now assume that $v$ has two children $L$ and $R$. 
At node $v$ if the follower goes to $L$, then $A_L[k]$ is the maximum utility the follower can obtain if the leader gets at least $k$.
If $A_L[k] < \mu_2(R)$, then the leader can not guarantee guarantee himself payoff of at least $k$ while follower plays $L$, since the follower guarantees strictly higher payoff of at least $\mu_2(R)$ by playing to $R$.
Similarly, if $A_R[k] < \mu_2(L)$, then the leader can not guarantee guarantee himself payoff of at least $k$ while follower plays $R$, since the follower guarantees strictly higher payoff of at least $\mu_2(L)$ by playing to $L$.
We consider 4 cases.  

(a) If $A_L[k] < \mu_2(R)$ and If $A_R[k] < \mu_2(L)$, then it is impossible for the leader to guarantee utility $k$, so $A_v[k] = -\infty$.  

(b) If $A_L[k] \ge \mu_2(R)$ and If $A_R[k] < \mu_2(L)$, then the leader can not incentivize the follower to play $R$ while guaranteeing utility $k$ for himself. However, he can play the strategy corresponding to $A_L[k]$ at $L$ and punishing strategy at $R$, which incentivizes the follower to play $R$ and guarantees utility of at least $k$ for the leader. Hence, $A_v[k] = A_L[k]$.

(c) If $A_L[k] < \mu_2(R)$ and If $A_R[k] \ge \mu_2(L)$, then with the similar reasoning, $A_v[k] = A_R[k]$. 

(d) If $A_L[k] \ge \mu_2(R)$ and If $A_R[k] \ge \mu_2(L)$, the leader can incentivize the follower to play either $L$ or $R$. So, $A_v[k] = \max\{A_L[k],A_R[k]\}$

In the case (d) if leader had access to a randomization device, it would be possible for the leader to incentivize the follower to play a mixed action at the node $v$, which may be further beneficial for the leader.
Instead in our setting, follower can use the memory states to imitate the mixed actions.  
For a given pair of leader and follower nodes, we define a procedure $achievable$ to check if this is possible. 

\textbf{Procedure $\bm{achievable}$}

For a leader node $v$, follower node $w$ where $v$ is an ancestor of $w$ and real number $U$,
we define a procedure $achievable(v,w,U)$ which determines if there exists a leader's strategy in which he randomizes at $v$, and a best response of the follower such that $w$ is reached with positive probability, conditional on the guarantee that the follower obtains the payoff of $U$ after reaching $w$. 
This procedure first creates a new game $\mathcal{G}_a$ as follows.  
The graph $G_a$ is constructed from the original graph $G$ by deleting all the nodes which are not ancestors of $w$ and adding one leaf node for each of the remaining follower nodes.
An edge is added from each of the follower node to the corresponding leaf node. 
The payoff for the follower at the leaf node corresponding to a follower node $v$ is $\max_{\text{child } v' \text{ of } v \text{ not in } G_a} \mu_2(v')$ and the payoff for the leader at that leaf node is 0. 

Now we perform an upward pass to compute for each node $x$ in $G_a$, the utility $\hat{u}_2(x)$ the follower can guarantee. To begin, we have $\hat{u}_2(w) = U$ and for leaf node $l$ with follower's utility $u$, let $\hat{u}_2(l) = u$. For a leader node $x$, let $\hat{u}_2(x) = \min \{\hat{u}_2(L),\hat{u}_2(R) | L,R \text{ children of } x\}$. For a follower node $x$, let $\hat{u}_2(x) = \max \{\hat{u}_2(L),\hat{u}_2(R) | L,R \text{ children of } x\}$. For a chance node $x$ which plays $L$ with probability $p$ and $R$ with probability $1-p$, $\hat{u}_2(x) = p \hat{u}_2(L) + (1-p) \hat{u}_2(R)$. 

Finally, we conclude that there exists such strategy for the leader, if $v$ and both of his children $L,R$ are in $G_a$ such that $\hat{u}_2(v) = \hat{u}_2(L) = \hat{u}_2(R) = U$.

\underline{Leader node:} If a leader node $v$ has only one child $w$, then $A_v[k] = A_w[k]$ for each $k$. 
Now assume that $v$ has two children $L$ and $R$. 
Leader can commit to payoff of $k$ in following two ways.

(i) The leader plays $L$ with probability $p$ and $R$ with probability $1-p$, the payoff of $k$ can be guaranteed by achieving payoff $i$ at the subgame rooted at $L$ and payoff $j$ at subgame rooted at $R$, while satisfying the inequality $pi+(1-p)j \ge k$.
Hence, the maximum utility for the follower in such a case, denoted by $X_v[k]$, can be calculated as the following maximization problem.     
\begin{equation}
X_v[k] = \max_{i, j, p} \: \set{p A_L[i] + (1- p) A_R[j] \hspace{2mm} \vert \hspace{2mm} p i + (1 - p) j \geq k}.
\label{eq:leadernode}
\end{equation}

(ii) Continuing the discussion of case (d) in the follower node, 
if leader had access to a randomization device, it would be possible to incentivize the follower to play a mixed action at the node $v$, which may be further beneficial for the leader.
Instead, leader can use memory for a randomization devise as follows.  

This can be done at a leader's node which is an ancestor to such a follower node in a following way.

For each of the follower node $w$ such that $v$ is ancestor of $w$, perform the procedure $achievable(v,w,X_v[k])$. For every such $w$ (with children $L_w$ and $R_w$) for which the procedure is successful, define 
\begin{equation}
Y_w[k] = \max_{i, j, p} \: \set{p A_{L_w}[i] + (1- p) A_{R_w}[j] \hspace{2mm} \vert \hspace{2mm} p i + (1 - p) j \geq k}.
\end{equation}

So, $Y_w[k]$ is a payoff that follower can get conditional on the leader gets at least $k$, when the follower chooses $L$ with probability $p$ and $R$ with probability $1-p$. 
The leader can incentivize the follower by choosing $L$ and $R$ himself at node $v$ with the same probabilities, as a suggestion that the follower should play $L$ at $w$ if he observes leader playing $L$ at $v$ and play $R$ otherwise.
If the follower does not comply, the leader commits to playing the punishing at subsequent nodes.
$Y_v[k] = \max_w \: \{X_w[k] \hspace{2mm} | \hspace{2mm} achievable(v,w,X_v[k]) \text{ is successful} \}$.
Finally, $A_v[k] = \max \{X_v[k], Y_v[k]\}$.
\qed

\section{Concluding remarks}\label{sec:conclude}

We study the two player, non-zero sum, perfect information sequential games on directed acyclic graphs and directed graphs. 
The players can play history-dependent, mixed behavioral strategies. 
We define an alternate formulation called strategy profiles with memory for the history-dependent strategy profiles in which players can base their decisions only on the memory states. 
Using this formulation, we establish that strategies with memory can be described efficiently if the memory size is polynomial.

While it was shown that it is NP-Hard  to compute the SSE for games on DAGs where the strategies are history-independent, we construct a polynomial time algorithm to compute the SSE for games on DAGs without chance nodes with strategy profiles with memory. 
Since the memory size of this strategy profile is linear in the number of nodes, it can be described very efficiently.
We modify our algorithm to work in games on general directed graphs without chance nodes, by proving an existence of SSE which visits each node at most twice. 
We also discuss the approximate version of SSE with memory for games on DAGs with chance nodes. 

We see many interesting extensions of our model and open question to be addressed in future work. 
Our algorithm and the construction of memory states can be used in various applied models in security games. 
One future extension would be to consider the games in which the leader and the follower play simultaneous moves in the dynamic games.

\section*{Acknowledgements}
This research was supported by the Czech Science Foundation (no. 19-24384Y) and by the Combat Capabilities Development Command Army Research Laboratory and was accomplished under Cooperative Agreement Number W911NF-13-2-0045 (ARL Cyber Security CRA). The views and conclusions contained in this document are those of the authors and should not be interpreted as representing the official policies, either expressed or implied, of the Combat Capabilities Development Command Army Research Laboratory or the U.S. Government. The U.S. Government is authorized to reproduce and distribute reprints for Government purposes not withstanding any copyright notation here on.

\bibliographystyle{splncs03}
\bibliography{ref}{}

\end{document}